\pdfoutput=1

\RequirePackage[l2tabu,orthodox]{nag}

\documentclass[11pt,letterpaper,reqno,oneside]{article}

\DeclareTextFontCommand{\emph}{\slshape}

\makeatletter
\renewcommand{\paragraph}{%
	\@startsection{paragraph}{4}%
	{\z@}{1.75ex \@plus 1ex \@minus .2ex}{-0.7em}%
	{\normalfont\normalsize\bfseries}%
}
\makeatother

\usepackage[T1]{fontenc}
\usepackage{lmodern}
\usepackage[utf8]{inputenc}

\usepackage[american,german,british]{babel}

\usepackage[activate={true,nocompatibility}]{microtype}

\usepackage{csquotes}

\usepackage[backend=biber,autolang=other,style=apa,maxcitenames=5,sortcites=false]{biblatex}
\DeclareLanguageMapping{british}{british-apa}
\DeclareLanguageMapping{american}{american-apa}

\usepackage{amsmath}
\usepackage{amssymb}
\usepackage{amsfonts}
\usepackage{amsthm}
\usepackage{mathtools}
\usepackage{stmaryrd}

\let\originalleft\left
\let\originalright\right
\renewcommand{\left}{\mathopen{}\mathclose\bgroup\originalleft}
\renewcommand{\right}{\aftergroup\egroup\originalright}

\usepackage{graphicx}
\usepackage{tikz,pgfplots}
\pgfplotsset{compat=1.10}
\usetikzlibrary{shapes.multipart}
\usetikzlibrary{decorations.markings}

\usepackage[title]{appendix}

\usepackage{datetime}
\newdateformat{datestyle}{ {\THEDAY} \monthname[\THEMONTH] \THEYEAR }

\usepackage{enumerate}
\usepackage{enumitem}
\setlist[enumerate,1]{label=(\arabic*)}
\setlist[itemize,1]{label=--}
\setlist[itemize,2]{label=--}
\setlist[itemize,3]{label=--}
\setlist[itemize,4]{label=--}

\usepackage{xcolor}
\usepackage{verbatim}
\usepackage{gensymb}
\usepackage{rotating}
\usepackage{subcaption}
\usepackage{floatpag}
\floatpagestyle{plain}
\usepackage{marvosym}

\usepackage{hyphenat}
\theoremstyle{definition}

\newtheorem{theorem}{Theorem}
\newtheorem{proposition}{Proposition}
\newtheorem{lemma}{Lemma}
\newtheorem{corollary}{Corollary}
\newtheorem{remark}{Remark}
\newtheorem{observation}{Observation}
\newtheorem{example}{Example}
\newtheorem{definition}{Definition}

\newtheoremstyle{named}
	{\topsep}					
	{\topsep}					
	{}							
	{0pt}						
	{\bfseries}					
	{}							
	{5pt plus 1pt minus 1pt}	
	{\thmnote{#3}}				
\theoremstyle{named}
\newtheorem{namedthm}{}

\usepackage{xpatch}
\xpatchcmd{\proof}{\itshape}{\proofheadfont}{}{}
\newcommand{\proofheadfont}{\slshape}

\usepackage{hyperref}
\hypersetup{pdfborder=0 0 0}

\usepackage[nameinlink]{cleveref}
\crefname{page}{p.}{pp.}
\crefname{equation}{equation}{equations}
\crefname{section}{section}{sections}
\crefname{subsection}{section}{sections}
\crefname{subsubsection}{section}{sections}
\crefname{appsec}{appendix}{appendices}
\crefname{supplsec}{supplemental appendix}{supplemental appendices}
\crefname{footnote}{footnote}{footnotes}
\crefname{figure}{figure}{figures}
\crefname{table}{table}{tables}
\crefname{theorem}{theorem}{theorems}
\crefname{proposition}{proposition}{propositions}
\crefname{lemma}{lemma}{lemmata}
\crefname{corollary}{corollary}{corollaries}
\crefname{remark}{remark}{remarks}
\crefname{observation}{observation}{observations}
\crefname{example}{example}{examples}
\crefname{fact}{fact}{facts}
\crefname{definition}{definition}{definitions}
\crefname{assumption}{assumption}{assumptions}
\crefname{exercise}{exercise}{exercises}
\crefname{notation}{notation}{notation}
\crefname{claim}{claim}{claims}
\crefname{conjecture}{conjecture}{conjectures}

\DeclareMathOperator*{\argmax}{arg\,max}

\DeclareMathOperator*{\co}{co}

\newcommand{\R}{\mathbf{R}}

\newcommand{\N}{\mathbf{N}}

\newcommand{\Z}{\mathbf{Z}}

\newcommand{\1}{\boldsymbol{1}}

\DeclarePairedDelimiter\abs{\lvert}{\rvert}

\newcommand*{\xslant}[2][76]{%
	\begingroup
	\sbox0{#2}%
	\pgfmathsetlengthmacro\wdslant{\the\wd0 + cos(#1)*\the\wd0}%
	\leavevmode
	\hbox to \wdslant{\hss
		\tikz[
			baseline=(X.base),
			inner sep=0pt,
			transform canvas={xslant=cos(#1)},
		] \node (X) {\usebox0};%
		\hss
		\vrule width 0pt height\ht0 depth\dp0 %
	}%
	\endgroup
}
\makeatletter
\newcommand*{\xslantmath}{}
\def\xslantmath#1#{%
	\@xslantmath{#1}%
}
\newcommand*{\@xslantmath}[2]{%
	\ensuremath{%
		\mathpalette{\@@xslantmath{#1}}{#2}%
	}%
}
\newcommand*{\@@xslantmath}[3]{%
	\xslant#1{$#2#3\m@th$}%
}
\makeatother

\makeatletter
\def\namedlabel#1#2{\begingroup
	#2%
	\def\@currentlabel{#2}%
	\phantomsection\label{#1}\endgroup
}
\makeatother

\makeatletter
\let\save@mathaccent\mathaccent
\newcommand*\if@single[3]{%
	\setbox0\hbox{${\mathaccent"0362{#1}}^H$}%
	\setbox2\hbox{${\mathaccent"0362{\kern0pt#1}}^H$}%
	\ifdim\ht0=\ht2 #3\else #2\fi
	}
\newcommand*\rel@kern[1]{\kern#1\dimexpr\macc@kerna}
\newcommand*\widebar[1]{\@ifnextchar^{{\wide@bar{#1}{0}}}{\wide@bar{#1}{1}}}
\newcommand*\wide@bar[2]{\if@single{#1}{\wide@bar@{#1}{#2}{1}}{\wide@bar@{#1}{#2}{2}}}
\newcommand*\wide@bar@[3]{%
	\begingroup
	\def\mathaccent##1##2{%
	  \let\mathaccent\save@mathaccent
	  \if#32 \let\macc@nucleus\first@char \fi
	  \setbox\z@\hbox{$\macc@style{\macc@nucleus}_{}$}%
	  \setbox\tw@\hbox{$\macc@style{\macc@nucleus}{}_{}$}%
	  \dimen@\wd\tw@
	  \advance\dimen@-\wd\z@
	  \divide\dimen@ 3
	  \@tempdima\wd\tw@
	  \advance\@tempdima-\scriptspace
	  \divide\@tempdima 10
	  \advance\dimen@-\@tempdima
	  \ifdim\dimen@>\z@ \dimen@0pt\fi
	  \rel@kern{0.6}\kern-\dimen@
	  \if#31
	    \overline{\rel@kern{-0.6}\kern\dimen@\macc@nucleus\rel@kern{0.4}\kern\dimen@}%
	    \advance\dimen@0.4\dimexpr\macc@kerna
	    \let\final@kern#2%
	    \ifdim\dimen@<\z@ \let\final@kern1\fi
	    \if\final@kern1 \kern-\dimen@\fi
	  \else
	    \overline{\rel@kern{-0.6}\kern\dimen@#1}%
	  \fi
	}%
	\macc@depth\@ne
	\let\math@bgroup\@empty \let\math@egroup\macc@set@skewchar
	\mathsurround\z@ \frozen@everymath{\mathgroup\macc@group\relax}%
	\macc@set@skewchar\relax
	\let\mathaccentV\macc@nested@a
	\if#31
	  \macc@nested@a\relax111{#1}%
	\else
	  \def\gobble@till@marker##1\endmarker{}%
	  \futurelet\first@char\gobble@till@marker#1\endmarker
	  \ifcat\noexpand\first@char A\else
	    \def\first@char{}%
	  \fi
	  \macc@nested@a\relax111{\first@char}%
	\fi
	\endgroup
}
\makeatother

\newcommand{\Situations}{\mathcal{S}}
\newcommand{\situations}{S}
\newcommand{\Objects}{\mathcal{X}}
\newcommand{\objects}{X}

\title{\scshape The comparative statics of dominance\thanks{We are grateful for comments from the editor, an associate editor, two referees, Ilia Krasikov, Giacomo Lanzani, Eddie Schlee, and a seminar audience at Arizona State University. Martín Grunwald provided excellent research assistance. Curello acknowledges support
from the German Research Foundation (DFG) through CRC TR 224 (Project B02).}}

\author{%
Gregorio Curello \\
Mannheim
\and
Ludvig Sinander \\
Oxford
\and
Mark Whitmeyer \\
Arizona State}

\date{12 August 2026}

\makeatletter
	\AtBeginDocument{ \hypersetup{
		pdftitle = {The comparative statics of dominance},
		pdfauthor = {Gregorio Curello, Ludvig Sinander and Mark Whitmeyer}
		} }
\makeatother

\begin{document}

\maketitle

\begin{abstract}
In finite problems comprising objects, cases, and an object- and case-contingent payoff function, we study the comparative statics of the set of undominated objects, meaning those for which there exists no mixture over objects that is superior in every case. We consider both weak and strict dominance (corresponding to different degrees of `strictness' in the definition of superiority). Our main theorem characterises those payoff transformations which robustly expand the not-weakly-dominated and not-strictly-dominated sets: the necessary and sufficient condition is that payoffs be transformed separately across cases, in either a monotone--concave or a constant manner. We apply our results to Pareto frontiers and games.
\end{abstract}

\section{Introduction}
\label{sec:intro}

Undominatedness is a fundamental solution concept in economics. Abstractly, in a finite problem comprising a finite set of \emph{objects,} a finite set of \emph{cases,} and an object- and case-dependent payoff function, the not-weakly-dominated (NWD) and not-strictly-dominated (NSD) sets comprise those objects which are not weakly/strictly dominated by any mixture over objects, where `weakly dominated' means weakly superior in every case and strictly superior in some, and `strictly dominated' means strictly superior in every case. In a given player's action-choice problem in a simultaneous-move game, the objects are her own actions, the cases are profiles of actions of the other players, and weak and strict dominance coincide with the familiar game-theoretic concepts bearing those names. In a social planner's problem, the objects are allocations, the cases are individuals, payoffs describe each individual's preferences over (lotteries over) allocations, and the NWD set is known as the \emph{Pareto frontier.}

In this paper, we study the comparative statics of dominance: which transformations of payoffs expand or contract the undominated sets? We are inspired by \textcite{Weinstein2016,BattigalliEtal2016}, who identified sufficient conditions for expansion of the NSD set; we seek conditions that are necessary and sufficient.

Following the comparative-statics literature \parencite[e.g.][]{Topkis1978,MilgromShannon1994,QuahStrulovici2009}, we formalise our questions in a robust, detail-free way, by asking which payoff transformations expand the NWD set in \emph{every} finite problem, and similarly for NSD and for contraction (rather than expansion) of the NWD and NSD sets.

Our main theorem provides the answer: a payoff transformation robustly contracts the NWD and NSD sets if and only if it is \emph{case-wise,} i.e. transforms payoffs separately case by case, and each case's transformation is convex and strictly increasing. Similarly, a payoff transformation robustly \emph{expands} the NWD (NSD) set if and only if it is case-wise and either (i)~each case's transformation is concave and strictly (weakly) increasing or (ii)~the transformation is degenerate in a precise sense to do with constant transformations. Economically, case-wiseness means that within-case payoffs are never contaminated by payoffs in other cases, and (by \hyperref[theorem:pratt]{Pratt's theorem},) convexity/concavity plus strict monotonicity means precisely decreased/increased risk-aversion in each case.

Our above definition of dominance is the standard one whereby an object is undominated only if no \emph{mixture} over objects dominates it. In an extension, we characterise the comparative statics of the sets of objects that are not (weakly or strictly) dominated by any \emph{object.} The conditions are weaker: monotonicity remains necessary, but concavity/convexity does not, and case-wiseness remains necessary for expansion but not for contraction.

Our first application is to Pareto dominance. Here, the objects $\objects$ are allocations, the cases $\{1,2,\dots,n\}$ are individuals, $U(\cdot,i) : \objects \to \R$ is the payoff function of individual $i$, and the not-weakly-dominated set is called the \emph{Pareto frontier.} Our main question in this application is which transformations of individuals' payoffs expand or contract the Pareto frontier in \emph{every} sub-economy (featuring only a subset of individuals). By our main theorem, the answer is (approximately): separate, individual-by-individual increases or decreases of risk-aversion.

Our second application is to strategic dominance in games. We fix a non-empty finite set $I$ of players, and consider all finite simultaneous-move games $(I,(A_i,u_i)_{i \in I})$ played by $I$. For a given player $i \in I$, the objects $A_i$ are her own actions, the cases $A_{-i}$ are profiles of her opponents' actions, her payoff function is $u_i : A_i \times A_{-i} \to \R$, and weak and strict dominance coincide with the usual game-theoretic concepts bearing those names. Our main theorem characterises those payoff transformations which robustly expand or contract the not-weakly-dominated and not-strictly-dominated sets. Applying this result repeatedly yields comparative statics for the set of rationalisable action profiles.

\subsection{Related literature}
\label{sec:intro:lit}

The classic literature on undominatedness, also known as `admissibility', spans game theory \parencite[since][]{Borel1921}\nocite{Borel1953}, statistical decision theory \parencite[since][]{Wald1939}, and social choice \parencite[since][]{Pareto1906}. Much of the focus of this literature is on identifying conditions under which every undominated object is `sometimes-optimal', i.e. expected-payoff-maximising against \emph{some} (subjective) probability distribution over cases \parencite[see][]{Wald1939,Stein1955,fishburn1975separation,Pearce1984,Borgers1993,Battigalli1997,BattigalliEtal2016,CheKimKojimaRyan2024,ChengBorgers2026}. We ask a different sort of question, concerning comparative statics.

Two beautiful theorems in \textcite{CheKimKojima2021} identify those payoff transformations which lead the NWD set to `shift up' (i.e. become `higher') in a set-wise sense---specifically, in the weak set order. Our comparative-statics question is orthogonal to theirs: what we ask is which transformations of payoffs lead the NWD and NSD sets to \emph{expand} (or contract).

Most closely related are \textcite{Weinstein2016,BattigalliEtal2016}, who consider broadly the same com\-pa\-ra\-tive-statics question as we do, focusing on strict dominance.%
    \footnote{See also \textcite[][section~8.1]{deOliveiraLamba2025}.}
What these authors show is that any case-independent concave and strictly increasing transformation of payoffs expands the NSD set (we recover this result in \cref{sec:appl:games} below). Our theorem extends and qualifies this finding, by identifying the entire class of payoff transformations which enjoy this expansion property, and also identifying those which enjoy the reverse `contraction' property. We furthermore answer the same questions for the NWD set.

At a higher level, we contribute to the comparative-statics literature \parencite[e.g.][]{Topkis1978,MilgromShannon1994,QuahStrulovici2009}. The most closely related paper here is \textcite{QuahStrulovici2012}, as we detail in a separate note \parencite{CurelloSinanderWhitmeyer2025note}.

This paper also contributes to the programme of Pratt's (\citeyear{Pratt1964}) theorem,%
    \footnote{Pratt's theorem is stated in \cref{sec:background:pratt} below.}
which seeks characterisations of the familiar `less risk-averse than' relation (which was first defined by \textcite{Yaari1969}). In particular, our result may be viewed as a new characterisation of `less risk-averse than' in terms of the undominated sets in different finite problems. Other recent work on `economic' characterisations of `less risk-averse than' includes \textcite{MaccheroniMarinacciWangWu2025,CurelloSinanderWhitmeyer2025}.

\subsection{Roadmap}
\label{sec:intro:roadmap}

We rehearse the relevant background in \cref{sec:background}. In \cref{sec:setup}, we introduce our setup. \Cref{sec:result} contains our main result, \Cref{thm:maintheorem}, characterising which transformations robustly expand or contract the undominated sets. The applications to Pareto dominance and to games are studied in \cref{sec:appl}. \Cref{sec:pure} contains a mirror theorem for dominance by \emph{objects} rather than by mixtures (\Cref{theorem:psd}). All proofs omitted from the main text may be found in the appendix.

\section{Background}
\label{sec:background}

In this section, we tersely recount some of the theory of comparative risk-aversion, as well as a theorem of Sierpiński on concave functions.

\subsection{Comparative risk aversion}
\label{sec:background:pratt}

Given a non-empty finite set $\objects$, we write $\Delta(\objects)$ for the set of all mixtures over $\objects$, i.e. all functions $p : \objects \to [0,1]$ which satisfy $\sum_{x \in \objects} p(x) = 1$. Recall that an expected-utility decision-maker is one for whom there exists a function $u : \objects \to \R$ such that she weakly prefers a mixture $p \in \Delta(\objects)$ to another mixture $q \in \Delta(\objects)$ if and only if $\sum_{x \in \objects} u(x) p(x) \geq \sum_{x \in \objects} u(x) q(x)$.

\begin{definition}[\cite{Yaari1969}]
For a non-empty finite set $\objects$ and functions $u,v : \objects \to \R$, we say that $v$ is \emph{more risk-averse than} $u$ if and only if for any $x \in \objects$ and $p \in \Delta(\objects)$, $\sum_{y \in \objects} v(y) p(y) \geq v(x)$ implies $\sum_{y \in \objects} u(y) p(y) \geq u(x)$ and $\sum_{y \in \objects} v(y) p(y) > v(x)$ implies $\sum_{y \in \objects} u(y) p(y) > u(x)$.
\end{definition}

Write `$\co A$' for the convex hull of a set $A \subseteq \R$. Recall Pratt's (\citeyear{Pratt1964}) theorem, which characterises `more risk-averse than':

\begin{namedthm}[Pratt's theorem.]
	\label{theorem:pratt}
	For a non-empty finite set $\objects$ and functions $u,v : \objects \to \R$, the following are equivalent:

	\begin{enumerate}[label=(\alph*)]
	
		\item \label{item:pratt:lra} $v$ is more risk-averse than $u$.

		\item \label{item:pratt:trans} 
        There exists a concave and strictly increasing function $\phi : \co(u(\objects)) \to \R$ such that $v(x) = \phi(u(x))$ for every $x \in \objects$.

		\item \label{item:pratt:curv} The following two properties hold:

		\begin{enumerate}[label=(\roman*),topsep=0em]
		
			\item \label{item:pratt:curv:ordequiv} For any $x,y \in \objects$, $u(x) \geq u(y)$ implies $v(x) \geq v(y)$ and $u(x) > u(y)$ implies $v(x) > v(y)$.

			\item \label{item:pratt:curv:compress} For any $x,y,z \in \objects$, if $u(x) < u(y) < u(z)$, then
			\begin{equation*}
				\frac{u(z)-u(y)}{u(y)-u(x)}
				\geq \frac{v(z)-v(y)}{v(y)-v(x)} .
			\end{equation*}
		
		\end{enumerate}
	
	\end{enumerate}
\end{namedthm}

Property \ref{item:pratt:curv}\ref{item:pratt:curv:compress} is a discrete version of the familiar ``pointwise ranking of Arrow--Pratt indices'' property.

\begin{remark}
	\label{remark:generalobjects}
	In the literature, comparative risk-aversion and its characterisation are almost only ever considered under the assumption that the objects are monetary: $\objects \subseteq \R$. However, by inspection, properties \ref{item:pratt:lra}--\ref{item:pratt:curv} are equally meaningful whatever the nature of the objects, and in fact they are equivalent even without the monetary assumption, as asserted above. See \textcite{CurelloSinanderWhitmeyer2025} for a proof.
\end{remark}


\subsection{Sierpiński's theorem on concave functions}
\label{sec:background:sierpinski}

Our proofs rely on the following beautiful result due to \textcite{Sierpinski1920}.

\begin{namedthm}[Sierpiński's theorem.]
	\label{theorem:sierpinski}
    Fix any non-empty open interval $I \subseteq \R$. A Lebesgue-measurable function $f : I \to \R$ is concave if and only if for all $u,v \in I$, $\frac{1}{2} f(u) + \frac{1}{2} f(v) \leq f\bigl(\frac{1}{2}u+\frac{1}{2}v\bigr)$.
\end{namedthm}

More particularly, our proofs rely on the specialisation of \hyperref[theorem:sierpinski]{Sierpiński's theorem} to weakly increasing $f$. For completeness, we include a self-contained proof of this specialisation in the appendix.

\section{Setup}
\label{sec:setup}

In this section, we introduce our environment, formalise the notion of a payoff transformation, and pose our comparative-statics question(s).

\subsection{Environment}
\label{sec:setup:setup}

There is a non-empty finite set $\Situations$ of \emph{cases}. A decision-maker faces a set $\objects$ of \emph{objects,} worries about some subset $\situations \subseteq \Situations$ of cases, and earns payoff $U(x,s) \in \R$ from object $x \in \objects$ in case $s \in \Situations$. Formally, we consider all \emph{finite problems,} meaning tuples $(\objects,\situations,U)$ where $\objects$ is a non-empty finite set, $\situations$ is a non-empty subset of $\Situations$, and $U$ is a map $\objects \times \Situations \to \R$.%
    \footnote{A technical detail is that `all finite problems' is not actually a set (in ZFC), since `all non-empty finite sets $\objects$' is not a set (in ZFC). The way around this annoyance is to fix throughout an (arbitrary) infinite set $\Objects$ of imaginable objects, and to restrict the definition of a finite problem $(\objects,\situations,U)$ by additionally requiring that $\objects \subseteq \Objects$. To save clutter, we leave all such `$\subseteq \Objects$' caveats implicit throughout the paper.}
We contemplate random choice, and assume expected utility, so that the payoff from mixture $p \in \Delta(\objects)$ in case $s \in \Situations$ is $\sum_{x \in \objects} U(x,s) p(x)$.%
    \footnote{Here $\Delta(\objects)$ denotes the set of all functions $p : \objects \to [0,1]$ satisfying $\sum_{x \in \objects} p(x) = 1$.}

The set of not-strictly-dominated objects in a finite problem $(\objects,\situations,U)$ is
\begin{align*}
	\text{NSD}(\objects,\situations,U) \coloneqq
	\bigl\{
	x \in \objects :
	{}&\text{there is no $p \in \Delta(\objects)$ such that}
	\\
	{}&\text{$\textstyle\sum_{y \in \objects} U(y,s) p(y) > U(x,s)$ for every $s \in \situations$}
	\bigr\} ,
\end{align*}
and the set of not-\emph{weakly}-dominated objects is
\begin{multline*}
	\text{NWD}(\objects,\situations,U) \coloneqq
		\bigl\{
		x \in \objects : \text{there is no $p \in \Delta(\objects)$ such that}
		\\
	\begin{aligned}
		{}&\text{$\textstyle\sum_{y \in \objects} U(y,s) p(y) \geq U(x,s)$ for every $s \in \situations$ and}
		\\
		{}&\text{$\textstyle\sum_{y \in \objects} U(y,s) p(y) > U(x,s)$ for some $s \in \situations$}
		\bigr\} .
	\end{aligned}
\end{multline*}

\begin{remark}
\label{remark:microfound}
The payoff function $U : \objects \times \Situations \to \R$ in a finite problem $(\objects,\situations,U)$ is most naturally interpreted as a composite $U = \mathcal{U} \circ Z$, where $\mathcal{U} : \mathcal{Z} \to \R$ describes the decision-maker's preferences over a non-empty set $\mathcal{Z}$ of `outcomes', and the function $Z : \objects \times \Situations \to \mathcal{Z}$ specifies which outcome prevails as a function of the object $x \in \objects$ and case $s \in \Situations$. Our formalism leaves the outcomes implicit, as is standard e.g. in game theory.
\end{remark}

The term `cases' is deliberately abstract, to fit a range of applications. As discussed in the introduction, the interpretation of `cases' varies by application, but the general idea is that a case specifies all factors which influence payoffs over objects. The set $\Situations$ contains all imaginable cases; depending on the application, it may be a very large set.

In a finite problem $(\objects,\situations,U)$, the set $\situations \subseteq \Situations$ comprises those cases which the decision-maker actually contemplates. In some applications, it is natural for the decision-maker always to contemplate all cases, meaning that only finite problems of the form $(\objects,\Situations,U)$ are relevant; we have results specifically about this scenario (see section~\ref{sec:result:props} below).

The restriction to \emph{finite} problems (with finitely many cases and finitely many objects) yields economically transparent results uncluttered by technical details, at the cost of ruling out some applications. The question of whether and how our analysis extends to infinite problems is left for future work.

For future reference, note that for each object $x \in \objects$ in a finite problem $(\objects,\situations,U)$, we may view $U(x,\cdot) : \Situations \to \R$ as a vector in $\R^\Situations$ (namely, the vector $\left( U(x,s) \right)_{s \in \Situations}$). As usual, for vectors $u,v \in \R^n$ (where $n \in \N$), `$u>v$'  means `$u \geq v$ and $u \neq v$', and $u \gg v$ means that $u_i > v_i$ for all $i \in \left\{1,\dots,n\right\}$.%
    \footnote{Hence $\text{NSD}(\objects,\Situations,U) = \{ x \in \objects : \text{$\nexists p \in \Delta(\objects)$ s.t. $\textstyle\sum_{y \in \objects} U(y,\cdot) p(y) \gg U(x,\cdot)$} \}$ and $\text{NWD}(\objects,\Situations,U) = \{ x \in \objects : \text{$\nexists p \in \Delta(\objects)$ s.t. $\textstyle\sum_{y \in \objects} U(y,\cdot) p(y) > U(x,\cdot)$} \}$.}

\subsection{Payoff transformations}
\label{sec:seteup:transformations}

For a given finite problem $(\objects,\situations,U)$, our question is which transformations of the payoff $U$ cause the undominated sets $\text{NSD}(\objects,\situations,U)$ and $\text{NWD}(\objects,\situations,U)$ to expand, and which ones cause them to contract. By a \emph{transformation,} we mean a map $\Psi : \R^{\objects \times \Situations} \to \R^{\objects \times \Situations}$, which associates each payoff function $U : \objects \times \Situations \to \R$ with a(nother) payoff function $[\Psi(U)] : \objects \times \Situations \to \R$.

\begin{namedthm}[\Cref{remark:microfound} {\normalfont (continued)}.]
\label{remark:microfound_shift}
When the payoff function $U : \objects \times \Situations \to \R$ is interpreted as a composite $U = \mathcal{U} \circ Z$, transformations $\Psi$ can capture both changes of tastes $\mathcal{U}$ and changes of the environment $Z$.
\end{namedthm}

We restrict attention to transformations $\Psi$ that treat objects \emph{neutrally} and \emph{separately.} The former property requires that for any payoff functions $U,V : \objects \times \Situations \to \R$ and any permutation $\pi : \objects \to \objects$ such that $U(x,\cdot) = V(\pi(x),\cdot)$ for all $x \in \objects$, we have $[\Psi(U)](x,\cdot) = [\Psi(V)](\pi(x),\cdot)$ for all $x \in \objects$. The latter property requires that for any $U,V : \objects \times \Situations \to \R$ and $x \in \objects$ such that $U(x,\cdot) = V(x,\cdot)$, we have $[\Psi(U)](x,\cdot) = [\Psi(V)](x,\cdot)$.

We view these two properties as natural and fairly weak. Many economic transformations satisfy them (some examples are given below). Extending our results to transformations which do not satisfy these properties might be interesting; we leave this for future work.

Recall that for each $x \in \objects$, we may view $U(x,\cdot)$ as a vector in $\R^\Situations$. It is easy to see that if a transformation $\Psi : \R^{\objects \times \Situations} \to \R^{\objects \times \Situations}$ treats objects neutrally and separately, then there exists a map $\Phi : \R^\Situations \to \R^\Situations$ such that $\Psi(U) = \Phi \circ U$ for any $U : \objects \times \Situations \to \R$. We henceforth identify each transformation $\Psi$ with its associated map $\Phi : \R^\Situations \to \R^\Situations$.

Each family $(\phi_s)_{s \in \Situations}$ of functions $\R \to \R$ generates a transformation $\Phi : \R^\Situations \to \R^\Situations$ via $(u_s)_{s \in \Situations} \mapsto (\phi_s(u_s))_{s \in \Situations}$. Such transformations are `case-wise': for each case $s \in \Situations$, the case-$s$ output $\phi_s(u_s) \in \R$ depends only on the case-$s$ input $u_s \in \R$. Formally, given a family $(\phi_s)_{s \in \Situations}$ of functions $\R \to \R$, we say that a transformation $\Phi : \R^\Situations \to \R^\Situations$ is \emph{$(\phi_s)_{s \in \Situations}$-case-wise} if and only if $\Phi((u_s)_{s \in \Situations}) = (\phi_s(u_s))_{s \in \Situations}$ for every $(u_s)_{s \in \Situations} \in \R^\Situations$. We say that a transformation $\Phi : \R^\Situations \to \R^\Situations$ is \emph{case-wise} if and only if it is $(\phi_s)_{s \in \Situations}$-case-wise for some family $(\phi_s)_{s \in \Situations}$ of functions $\R \to \R$.

\begin{example}
\label{example:expand}
Let $(\phi_s)_{s \in \Situations}$ be a family of concave and strictly increasing functions $\R \to \R$, and let $\Phi : \R^\Situations \to \R^\Situations$ be $(\phi_s)_{s \in \Situations}$-case-wise. By \hyperref[theorem:pratt]{Pratt's theorem}, this transformation increases risk-aversion in every case.
\end{example}

\begin{example}
\label{example:simple}
\textcite{Weinstein2016,BattigalliEtal2016} considered the specialisation of \Cref{example:expand} in which risk-aversion is increased \emph{in the same way} in each case: there is a single concave and strictly increasing function $\phi : \R \to \R$, and $\Phi : \R^\Situations \to \R^\Situations$ is given by $\Phi((u_s)_{s \in \Situations}) = (\phi(u_s))_{s \in \Situations}$ for every $(u_s)_{s \in \Situations} \in \R^\Situations$.
\end{example}

\begin{example}
\label{example:cubic}
Let $\phi : \R \to \R$ be strictly increasing but neither concave nor convex (e.g. $u \mapsto u^3$), and let $\Phi : \R^\Situations \to \R^\Situations$ be given by $\Phi((u_s)_{s \in \Situations}) = (\phi(u_s))_{s \in \Situations}$ for every $(u_s)_{s \in \Situations} \in \R^\Situations$. By \hyperref[theorem:pratt]{Pratt's theorem}, this transformation neither increases nor decreases risk-aversion.
\end{example}

\begin{example}
\label{example:insurance}
Fix a probability distribution $\mu \in \Delta(\Situations)$, and let $\Phi : \R^\Situations \to \R^\Situations$ be given by $\Phi((u_s)_{s \in \Situations}) = \left( \sum_{s \in \Situations} u_s \mu(s), \dots, \sum_{s \in \Situations} u_s \mu(s) \right)$ for every $(u_s)_{s \in \Situations} \in \R^\Situations$. This transformation provides (full) insurance: in any finite problem $(\objects,\situations,U)$, each object $x \in \objects$ now earns $\sum_{s \in \Situations} U(x,s) \mu(s)$ in every case. Obviously this transformation is not case-wise.
\end{example}

\subsection{Comparative-statics question(s)}
\label{sec:seteup:questions}

Our question is which payoff transformations $\Phi : \R^\Situations \to \R^\Situations$ cause the undominated sets to expand or contract. We seek a `robust' answer to this question, which does not depend on which particular finite problem is considered. Concretely, we are interested in the following four properties which a given transformation $\Phi : \R^\Situations \to \R^\Situations$ may or may not enjoy:

\begin{enumerate}
\item \label{item:maintheorem1} $\text{NSD}(\objects,\situations,U) \subseteq \text{NSD}(\objects,\situations,\Phi \circ U)$ for every finite problem $(\objects,\situations,U)$.
\item \label{item:maintheorem2} $\text{NWD}(\objects,\situations,U) \subseteq \text{NWD}(\objects,\situations,\Phi \circ U)$ for every finite problem $(\objects,\situations,U)$.
\item \label{item:shrink1} $\text{NSD}(\objects,\situations,\Phi \circ U) \subseteq \text{NSD}(\objects,\situations,U)$ for every finite problem $(\objects,\situations,U)$.
\item \label{item:shrink2} $\text{NWD}(\objects,\situations,\Phi \circ U) \subseteq \text{NWD}(\objects,\situations,U)$ for every finite problem $(\objects,\situations,U)$.
\end{enumerate}
For each of these four properties, our question is which transformations $\Phi$ enjoy that property.

By asking `robust' questions, requiring comparative statics to hold across all finite problems, we obtain answers which are independent of the details of any given problem. The answers to our `robust' questions will furthermore turn out to be simple, providing straightforward sufficient conditions for comparative statics which are as weak as possible (subject to `robustness').

In more detail, the advantage of considering various sets $\objects$ of objects is that this yields conclusions which are robust to any constraints which the decision-maker may face. Similarly, considering different sets $\situations \subseteq \Situations$ of cases yields conclusions that are robust to the decision-maker having a particular model or world-view, in which only certain cases are contemplated. For example, the arrival of new information may allow the decision-maker to rule out certain cases.

In some applications, ruling out cases is not natural or well-motivated. For such situations, we will have some results below that are `partially robust', considering only finite problems $(\objects,\situations,U)$ such that $\situations=\Situations$.

\section{Main theorem}
\label{sec:result}

\begin{theorem}
	\label{thm:maintheorem}
    Assume that $2 \leq \abs*{\Situations} < \infty$, and consider a transformation $\Phi : \R^\Situations \to \R^\Situations$.

	\begin{enumerate}[label=(\alph*)]
	
		\item \label{dom_noncomponent:nsd} Property~\ref{item:maintheorem1} holds if and only if there exists a family $(\phi_s)_{s \in \Situations}$ of weakly increasing functions $\R \to \R$ such that
        $\Phi$ is $(\phi_s)_{s \in \Situations}$-case-wise and either (i)~$\phi_s$ is constant for all but at most one $s \in \Situations$ or (ii)~$\phi_s$ is concave for all $s \in \Situations$.

		\item \label{dom_noncomponent:nwd} Property~\ref{item:maintheorem2} holds if and only if there exists a family $(\phi_s)_{s \in \Situations}$ of functions $\R \to \R$ such that $\Phi$ is $(\phi_s)_{s \in \Situations}$-case-wise and either (i)~$\phi_s$ is concave and strictly increasing for each $s \in \Situations$ or (ii)~$\phi_s$ is constant for every $s \in \Situations$.

        \item \label{dom_noncomponent:grow} 
        Property~\ref{item:shrink1} holds if and only if Property~\ref{item:shrink2} holds, which holds if and only if there exists a family $(\phi_s)_{s \in \Situations}$ of convex and strictly increasing functions $\R \to \R$ such that $\Phi$ is $(\phi_s)_{s \in \Situations}$-case-wise.
	
	\end{enumerate}
\end{theorem}

\Cref{thm:maintheorem} says that case-wiseness is necessary for a transformation $\Phi$ to robustly expand or robustly contract the undominated sets, that a monotonicity-plus-concavity/convexity property is also necessary for robust expansion/contraction, and that these necessary conditions are jointly sufficient. One implication is that the transformations in \Cref{example:expand,example:simple} robustly expand the undominated sets. The transformation in \Cref{example:cubic} neither robustly expands nor robustly contracts the undominated sets, since it neither increases nor decreases risk-aversion. The insurance-type transformation in \Cref{example:insurance} also neither robustly expands nor robustly contracts the undominated sets, since it is not case-wise.

The proofs of the three parts \ref{dom_noncomponent:nsd}--\ref{dom_noncomponent:grow} are similar to each other. Sufficiency follows from a short elegant argument substantially due to \textcite{Weinstein2016,BattigalliEtal2016} (although these authors use it to establish a weaker claim). Focusing on part~\ref{dom_noncomponent:grow}, this argument goes as follows. The condition in \ref{dom_noncomponent:grow} implies, by \hyperref[theorem:pratt]{Pratt's theorem}, that applying $\Phi$ makes the decision-maker less risk-averse in every case, meaning (by definition) that whenever a mixture $p \in \Delta(\objects)$ is (strictly) preferred to an object $x \in \objects$ in a case $s \in \situations$ of a finite problem $(\objects,\situations,U)$, the same is true in the transformed problem $(\objects,\situations,\Phi \circ U)$. Thus any object $x \in \objects$ which was dominated (by some mixture $p \in \Delta(\objects)$) in $(\objects,\situations,U)$ remains dominated (by that same mixture $p$) in $(\objects,\situations,\Phi \circ U)$; in other words, applying $\Phi$ expands the dominated set, so shrinks the \emph{un}dominated set.

Necessity is more difficult. For each of the three parts \ref{dom_noncomponent:nsd}--\ref{dom_noncomponent:grow}, we proceed in two steps, first showing that the transformation $\Phi$ must be case-wise, and then that the functions $(\phi_s)_{s \in \Situations}$ must satisfy the stated properties.

In the first step, there is a perhaps surprising asymmetry between expansion (parts~\ref{dom_noncomponent:nsd} and \ref{dom_noncomponent:nwd}) and contraction (part~\ref{dom_noncomponent:grow}). It is almost immediate that the expansive properties~\ref{item:maintheorem1} and \ref{item:maintheorem2} necessitate case-wiseness.%
\footnote{\label{footnote:sitwise_nec}For any case $s \in \Situations$, whenever $u,v \in \R^\Situations$ satisfy $u_s = v_s$, it must be that $\Phi_s(u) = \Phi_s(v)$, since otherwise letting $x \neq y$, $U(x,\cdot) \coloneqq u$ and $U(y,\cdot) \coloneqq v$, we have $\text{NSD}(\{x,y\},\{s\},U) = \{x,y\} \nsubseteq \text{NSD}(\{x,y\},\{s\},\Phi \circ U)$, and similarly for $\text{NWD}$.}

By contrast, it takes some work to show that the contractive properties~\ref{item:shrink1} and \ref{item:shrink2} necessitate case-wiseness. It is easily shown that these properties necessitate \emph{strict case-wise monotonicity:} for any $u,v \in \R^\Situations$ and any case $s \in \Situations$, $u_s > v_s$ implies $\Phi_s(u) > \Phi_s(v)$.%
    \footnote{If $\Phi_s(u) \leq \Phi_s(v)$ and $u_s > v_s$, then fixing $x \neq y$ and letting $U(x,\cdot) \coloneqq v$ and $U(y,\cdot) \coloneqq u$ yields $\text{NSD}(\{x,y\},\{s\},U) \not\ni x \in \text{NSD}(\{x,y\},\{s\},\Phi \circ U)$, and similarly $\text{NWD}$.}
We then show that for each $s \in \Situations$, the map $u_s \mapsto \Phi_s\left((u_t)_{t \in \Situations}\right)$ must be continuous; this and strict case-wise monotonicity together imply case-wiseness. To establish continuity, we show that any discontinuity would permit the construction of a finite problem $(\{x,y,z\},\Situations,U)$ that violates properties~\ref{item:shrink1} and \ref{item:shrink2}.

\begin{remark}
\label{remark:rich}
This continuity argument relies on the richness of the payoff domain. If we were to only contemplate (say) integer payoffs, i.e. if $\Phi$ were a map $\Phi : \Z^\Situations \to \Z^\Situations$, then the continuity argument would fail, and case-wiseness would in fact \emph{not} be necessary for properties~\ref{item:shrink1} and \ref{item:shrink2}, as we show in \cref{app:pf_remark:rich}.
\end{remark}

Having shown that properties~\ref{item:maintheorem1}--\ref{item:shrink2} each imply the existence of a family $(\phi_s)_{s \in \Situations}$ of functions $\R \to \R$ such that the transformation $\Phi$ is $(\phi_s)_{s \in \Situations}$-case-wise, it remains to establish that $(\phi_s)_{s \in \Situations}$ must have the properties asserted by parts~\ref{dom_noncomponent:nsd}--\ref{dom_noncomponent:grow}. This is accomplished by a pair of propositions given in the next section.

\subsection{Case-wise transformations}
\label{sec:result:props}

We now identify the conditions on \emph{case-wise} transformations that are necessary and sufficient for comparative statics. These results form part of the proof of \Cref{thm:maintheorem}, but are also of independent interest.

We characterise not only the fully robust `for all finite problems' com\-pa\-ra\-tive-statics properties~\ref{item:maintheorem1}--\ref{item:shrink2}, but also the following four `partially robust' comparative-statics properties in which the full set $\Situations$ of cases (rather than a subset $\situations \subseteq \Situations$) is always contemplated:

\begin{enumerate}[label=(\arabic*{*})]
\item \label{item:NSDcomponent} $\text{NSD}(\objects,\Situations,U) \subseteq \text{NSD}(\objects,\Situations,\Phi \circ U)$ for every finite problem $(\objects,\Situations,U)$.
\item \label{item:NWDcomponent} $\text{NWD}(\objects,\Situations,U) \subseteq \text{NWD}(\objects,\Situations,\Phi \circ U)$ for every finite problem $(\objects,\Situations,U)$.
\item \label{item:NSDcomponentshrink} $\text{NSD}(\objects,\Situations,\Phi \circ U) \subseteq \text{NSD}(\objects,\Situations,U)$ for every finite problem $(\objects,\Situations,U)$.
\item \label{item:NWDcomponentshrink} $\text{NWD}(\objects,\Situations,\Phi \circ U) \subseteq \text{NWD}(\objects,\Situations,U)$ for every finite problem $(\objects,\Situations,U)$.
\end{enumerate}
Characterising properties \ref{item:NSDcomponent}--\ref{item:NWDcomponentshrink} is valuable for applications, as we shall see in \cref{sec:appl} below.

Our first proposition concerns only expansions of undominated sets.

\begin{proposition}
	\label{proposition:dom_component}
    Assume that $2 \leq \abs*{\Situations} < \infty$. Fix a family $(\phi_s)_{s \in \Situations}$ of functions $\R \to \R$, and let $\Phi : \R^\Situations \to \R^\Situations$ be $(\phi_s)_{s \in \Situations}$-case-wise.

	\begin{enumerate}[label=(\alph*)]

        \item \label{dom_component:nsd_main} Property~\ref{item:maintheorem1} holds if and only if either (i)~$\phi_s$ is concave and weakly increasing for each $s \in \Situations$ or (ii)~$\phi_s$ is weakly increasing for some $s \in \Situations$ and constant for all other $s \in \Situations$.
	
		\item \label{dom_component:nsd} Property~\ref{item:NSDcomponent} holds if and only if either (i)~$\phi_s$ is concave and weakly increasing for each $s \in \Situations$ or (ii)~$\phi_s$ is constant for some $s \in \Situations$.

		\item \label{dom_component:nwd} Property~\ref{item:maintheorem2} holds if and only if property~\ref{item:NWDcomponent} holds, which holds if and only if either (i)~$\phi_s$ is concave and strictly increasing for each $s \in \Situations$ or (ii)~$\phi_s$ is constant for every $s \in \Situations$.
        
	\end{enumerate}
\end{proposition}

The core logic behind \Cref{proposition:dom_component} is most easily understood by focusing on the following specialisation, in which only case-independent transformations ($\phi_s = \phi$ for every $s \in \Situations$) and the fully robust properties~\ref{item:maintheorem1} and \ref{item:maintheorem2} are considered.

\begin{corollary}
	\label{corr:dom_component_single}
    Assume that $2 \leq \abs*{\Situations} < \infty$. Fix a function $\phi : \R \to \R$, and let $\Phi : \R^\Situations \to \R^\Situations$ be given by $\Phi((u_s)_{s \in \Situations}) = (\phi(u_s))_{s \in \Situations}$ for every $(u_s)_{s \in \Situations} \in \R^\Situations$.

    \begin{itemize}

        \item Property~\ref{item:maintheorem1} holds if and only if $\phi$ is concave and weakly increasing.

        \item \label{dom_component:nwd_single} Property~\ref{item:maintheorem2} holds if and only if either (i)~$\phi$ is concave and strictly increasing or (ii)~$\phi$ is constant.

    \end{itemize}
\end{corollary}

The sufficiency (`if') claims in \Cref{corr:dom_component_single} follow from the `short elegant argument' described immediately after \Cref{thm:maintheorem} above, except when $\phi$ is constant: in that scenario, the undominated sets trivially expand when $\Phi$ is applied because $\text{NSD}(\objects,\situations,\Phi \circ U) = \objects = \text{NWD}(\objects,\situations,\Phi \circ U)$ for every finite problem $(\objects,\situations,U)$.

\begin{proof}[Proof of necessity (`only if') in \Cref{corr:dom_component_single}]
We first show the necessity of the monotonicity properties, then of concavity. Fix arbitrary $x \neq y \neq z \neq x$. Further fix two cases $s \neq t$ in $\Situations$ (this is possible since $\abs*{\Situations} \geq 2$).

Suppose that $\phi$ is not weakly increasing, i.e. there are $u>v$ such that $\phi(u) < \phi(v)$. Then for $U(x,s) \coloneqq u$ and $U(y,s) \coloneqq v$, it holds that $\text{NSD}(\{x,y\},\{s\},U) \ni x \notin \text{NSD}(\{x,y\},\{s\},\Phi \circ U)$, so property~\ref{item:maintheorem1} fails.

Suppose that $\phi$ is neither strictly increasing nor constant, i.e. there are $u>v$ and $u' \neq v'$ such that $\phi(u) \leq \phi(v)$ and $\phi(u') < \phi(v')$. Then for
\begin{equation*}
    U\left(\text{row},\text{column}\right)
    \coloneqq
    \begin{tabular}{r|cc}
          & $s$ & $t$ \\ \hline
        $x$ & $u$ & $u'$ \\
        $y$ & $v$ & $v'$
    \end{tabular} 
\end{equation*}
we have $\text{NWD}(\{x,y\},\{s,t\},U) \ni x \notin \text{NWD}(\{x,y\},\{s,t\},\Phi \circ U)$, so property~\ref{item:maintheorem2} fails.

Assume for the remainder that $\phi$ is weakly increasing. Suppose that $\phi$ is not concave. Then by \hyperref[theorem:sierpinski]{Sierpiński's theorem} from \cref{sec:background:sierpinski} above (applicable since increasing functions are Lebesgue-measurable), there exist $u>v$ such that $\frac{1}{2} \phi(u) + \frac{1}{2} \phi(v) > \phi\bigl( \frac{1}{2} u + \frac{1}{2} v \bigr)$. Hence for
\begin{equation*}
    U\left(\text{row},\text{column}\right)
    \coloneqq
    \begin{tabular}{r|ccc}
          & $s$ & $t$ \\ \hline
        $x$ & $u$ & $v$ \\
        $y$ & $v$ & $u$ \\
        $z$ & $\frac{1}{2} u + \frac{1}{2} v$ & $\frac{1}{2} u + \frac{1}{2} v$
    \end{tabular} 
\end{equation*}
we have $z \in \text{NWD}(\{x,y,z\},\{s,t\},U) \subseteq \text{NSD}(\{x,y,z\},\{s,t\},U)$, and keeping in mind the mixture $p \in \Delta(\{x,y,z\})$ with $p(x) \coloneqq 1/2 \eqqcolon p(y)$, we see that $z \notin \text{NSD}(\{x,y,z\},\{s,t\},\Phi \circ U) \supseteq \text{NWD}(\{x,y,z\},\{s,t\},\Phi \circ U)$. Thus properties~\ref{item:maintheorem1} and \ref{item:maintheorem2} both fail.
\end{proof}

Our second proposition concerns only \emph{contractions} of undominated sets.

\begin{proposition}
	\label{proposition:dom_component2}
    Assume that $2 \leq \abs*{\Situations} < \infty$. Fix a family $(\phi_s)_{s \in \Situations}$ of functions $\R \to \R$, and let $\Phi : \R^\Situations \to \R^\Situations$ be $(\phi_s)_{s \in \Situations}$-case-wise. Property~\ref{item:shrink1} holds if and only if property~\ref{item:NSDcomponentshrink} holds, which holds if and only if property~\ref{item:shrink2} holds, which holds if and only if property~\ref{item:NWDcomponentshrink} holds, which holds if and only if $\phi_s$ is convex and strictly increasing for each $s \in \Situations$.
\end{proposition}

\Cref{proposition:dom_component2} is similar to \Cref{proposition:dom_component}, and is true for similar reasons. Strict monotonicity is required in \Cref{proposition:dom_component2} since any weakly-but-not-strictly increasing transformation sometimes makes a previously dominated object undominated.

\begin{remark}
By \hyperref[theorem:pratt]{Pratt's theorem}, the condition in \Cref{proposition:dom_component2} is equivalent to reducing risk-aversion in every case. By another theorem of ours \parencite[Theorem~1 in][]{CurelloSinanderWhitmeyer2025}, reduced risk-aversion is behaviourally equivalent to the decision-maker gaining access to a (possibly random) outside option. It follows that dominated choices (in a decision problem or a game) can never be explained by the presence of an unobserved outside option: any choice that is dominated in the decision-maker's problem exclusive of the outside option, $(\objects,\situations,U)$, is also dominated in the `true' problem $(\objects,\situations,\Phi \circ U)$ that accounts for the outside option.
\end{remark}

\section{Applications}
\label{sec:appl}

In this section, we apply our results to Pareto frontiers and to undominated (and rationalisable) actions in games.

\subsection{Pareto dominance}
\label{sec:appl:pareto}

Fix a population $\Situations \coloneqq \{1,2,\dots,n\}$ of $n \geq 2$ individuals. An \emph{economy} is a pair $(\objects,U)$, where $\objects$ is a non-empty finite set of \emph{allocations} and $U : \objects \times \Situations \to \R$ encodes payoffs: $U(\cdot,i) : \objects \to \R$ is the payoff function of individual $i \in \Situations$. (The restriction to finite sets $\objects$ of allocations rules out some applications.) We assume expected utility, so that in an economy $(\objects,U)$, the payoff of individual $i \in \Situations$ from a mixture $p \in \Delta(\objects)$ is $\sum_{x \in \objects} U(x,i) p(x)$.

The \emph{Pareto frontier} $\text{PF}_\Situations(\objects,U)$ of an economy $(\objects,U)$ is the set of all allocations such that no mixture is weakly preferred by all individuals and strictly preferred by some; in other words, $\text{PF}_\Situations(\objects,U) \coloneqq \text{NWD}(\objects,\Situations,U)$. More generally, when only a (non-empty) sub-population $\situations \subseteq \Situations$ is considered, the Pareto frontier of economy $(\objects,U)$ is $\text{PF}_\situations(\objects,U) \coloneqq \text{NWD}(\objects,\situations,U)$.

\begin{remark}
	\label{remark:pareto_pearce}
    For any economy $(\objects,U)$ and sub-population $\situations \subseteq \Situations$, by the `cautious' version of Pearce's lemma \parencite[see e.g.][Theorem~1.7]{Myerson1991}, an allocation $x \in \objects$ belongs to $\text{PF}_\situations(\objects,U)$ if and only if there exist strictly positive weights $(\lambda_i)_{i \in \situations}$ such that $x \in \argmax_{y \in \objects} \sum_{i \in \situations} \lambda_i U(y,i)$.%
		\footnote{Curiously, this characterisation of $\text{NWD}$ does not extend straightforwardly to infinite sets $\objects$ of objects/allocations \parencite[see][]{CheKimKojimaRyan2024}. By contrast, the familiar (`incautious') Pearce's lemma for $\text{NSD}$
        \parencite[see e.g.][Theorem~1.6]{Myerson1991} extends directly to infinite $\objects$ given some compactness and upper semi-continuity.}
    In other words, membership of the Pareto frontier is equivalent to being the optimal choice of a utilitarian planner, for some Pareto weights. This motivates interest in the Pareto frontier as defined above (the \emph{deterministic} allocations that are not dominated by any \emph{mixture} over allocations). One could of course also consider an alternative definition in which only dominance by \emph{allocations} is contemplated; we do this in \cref{sec:pure} below.
\end{remark}

\Cref{thm:maintheorem} characterises which payoff transformations generate expansion or contraction of the Pareto frontier $\text{PF}_\situations(\objects,U)$ of every economy $(\objects,U)$ for every sub-population $\situations \subseteq \Situations$. One upshot is that such transformations must be individual-by-individual, with no scope for one individual's preferences to influence how another individual's preferences are transformed.

This has some striking implications. For example, it is natural to conjecture that increased alignment of interests among individuals in the population will contract the Pareto frontier. \Cref{thm:maintheorem} implies that this conjecture is false, at least when `increased alignment' is formalised as payoffs $U : \objects \times \Situations \to \R$ shifting to $\Phi \circ U : \objects \times \Situations \to \R$ where $\Phi : \R^\Situations \to \R^\Situations$ is weakly increasing, so that (ceteris paribus) each individual likes an object $x \in \objects$ more the more it is liked by others. Indeed, case-wiseness requires that $u_s \mapsto \Phi_t((u_r)_{r \in \Situations})$ be constant for all $s \neq t$ in $\Situations$, meaning that alignment remains unchanged (neither strictly increases nor strictly decreases).

The other upshot of \Cref{thm:maintheorem} is that (by \hyperref[theorem:pratt]{Pratt's theorem}) contraction requires each individual to become less risk-averse, while (neglecting some exceptional scenarios) expansion requires that each individual become more risk-averse.

Depending on the application, the population may simply be fixed, making it unnatural to consider sub-populations. \Cref{proposition:dom_component,proposition:dom_component2} characterise, among individual-by-individual transformations, which ones expand or contract the full-population Pareto frontier $\text{PF}_\Situations(\objects,U)$ of every economy $(\objects,U)$. Again, the gist of the conditions is that expansion (contraction) requires individuals to become more (less) risk-averse.

\subsection{Strategic dominance in games}
\label{sec:appl:games}

Fix a non-empty set $I$ of players, with $\abs*{I} \geq 2$. We consider all finite simultaneous-move games played by $I$, meaning all tuples $(I,(A_i,u_i)_{i \in I})$ such that for each $i \in I$, $A_i$ is a non-empty finite set and $u_i$ is a map $A \to \R$, where $A \coloneqq \prod_{j \in I} A_j$. The interpretation is that players simultaneously choose actions, and that for each player~$i \in I$, her available actions are $A_i$ and her payoff is $u_i(a)$, where $a \in A$ is the profile of all players' actions. Weak and strict dominance coincide with the familiar game-theoretic concepts bearing those names: player~$i$'s not-weakly-dominated actions are $\text{NWD}(A_i,A_{-i},u_i)$, and her not-strictly-dominated actions are $\text{NSD}(A_i,A_{-i},u_i)$, where $A_{-i} \coloneqq \prod_{j \in I \setminus \{i\}} A_j$.

By \Cref{thm:maintheorem}, the transformations $\Phi_i$ of player~$i$'s payoff which contract or expand $\text{NWD}(A_i,A_{-i},u_i)$ or $\text{NSD}(A_i,A_{-i},u_i)$ in every game $(I,(A_i,u_i)_{i \in I})$ must transform player~$i$'s payoffs profile-by-profile: for any profile $a_{-i} \in A_{-i}$ of opponents' actions, player~$i$'s transformed payoffs $(\Phi_i \circ u_i)(\cdot,a_{-i}) : A_i \to \R$ against a profile $a_{-i} \in A_{-i}$ can depend only on the pre-transformation payoffs $u_i(\cdot,a_{-i})$ \emph{against that same profile $a_{-i}$,} not on the pre-transformation payoffs $u_i(\cdot,b_{-i})$ against any other profile $b_{-i} \in A_{-i} \setminus \{a_{-i}\}$. In terms of the payoff matrix of a two-player game, for the row player, this means a column-by-column transformation, with the transformed entries in each column depending only on the pre-transformed entries in that same column.

\Cref{thm:maintheorem} furthermore asserts that if a transformation is to contract player~$i$'s undominated sets, then (by \hyperref[theorem:pratt]{Pratt's theorem}) it must be that against each profile $a_{-i} \in A_{-i}$ of opponents' actions, player~$i$'s preference over (mixtures over) her own actions becomes less risk-averse. Similarly, expansion requires---degenerate transformations aside---that player~$i$ become more risk-averse, profile by profile.

For each player $i \in I$ and integer $n \in \N$, let $A_i^n \coloneqq \text{NSD}(A_i^{n-1},A_{-i}^{n-1},u_i)$, where for each $i \in I$, $A_{-i}^n \coloneqq \prod_{j \in I \setminus \{i\}} A_j^n$ for all $n \in \N$, $A_i^0 \coloneqq A_i$ and $A_{-i}^0 \coloneqq A_{-i}$. The rationalisable action profiles are exactly $\prod_{i \in I} \bigcap_{n \in \N} A_i^n$. By applying \Cref{thm:maintheorem} iteratively for each $n \in \N$, we see that when all players become more (less) risk-averse in a profile-by-profile fashion, the set of rationalisable action profiles expands (shrinks).

A special scenario is when player~$i$'s payoff is transformed in a profile-independent way: that is, $\Phi_i$ is not only $(\phi_i^{a_{-i}})_{a_{-i} \in A_{-i}}$-case-wise, but there is a single function $\phi_i : \R \to \R$ such that $\phi_i^{a_{-i}} = \phi_i$ for every profile $a_{-i} \in A_{-i}$. In this special scenario, our comparative-statics questions become:

\begin{enumerate}[label=(\arabic*$^\dag$)]

	\item \label{item:NSD} $\text{NSD}(A_i,A_{-i},u_i) \subseteq \text{NSD}(A_i,A_{-i},\phi_i \circ u_i)$ for every player $i \in I$ and every finite game $(I,(A_j,u_j)_{j \in I})$.

	\item \label{item:NWD} $\text{NWD}(A_i,A_{-i},u_i) \subseteq \text{NWD}(A_i,A_{-i},\phi_i \circ u_i)$ for every player $i \in I$ and every finite game $(I,(A_j,u_j)_{j \in I})$.

    \item \label{item:NSDflip} $\text{NSD}(A_i,A_{-i},\phi_i \circ u_i) \subseteq \text{NSD}(A_i,A_{-i},u_i)$ for every player $i \in I$ and every finite game $(I,(A_j,u_j)_{j \in I})$.

	\item \label{item:NWDflip} $\text{NWD}(A_i,A_{-i},\phi_i \circ u_i) \subseteq \text{NWD}(A_i,A_{-i},u_i)$ for every player $i \in I$ and every finite game $(I,(A_j,u_j)_{j \in I})$.

\end{enumerate}
The following corollary of \Cref{proposition:dom_component,proposition:dom_component2} gives necessary and sufficient conditions for comparative statics in this special scenario. The `if' half of part~\ref{item:NSD} recovers the result of \textcite{Weinstein2016,BattigalliEtal2016}; the rest is (presumably) new.

\begin{corollary}
	\label{corollary:dom_mcs_charac}
	Assume that $2 \leq \abs*{I} < \infty$, and let $(\phi_i)_{i \in I}$ be maps $\R \to \R$.

	\begin{enumerate}[label=(\alph*)]
	
		\item \label{dom_mcs_charac:nsd} Property~\ref{item:NSD} holds if and only if for each player $i \in I$, $\phi_i$ is concave and weakly increasing.

		\item \label{dom_mcs_charac:nwd} Property~\ref{item:NWD} holds if and only if for each player $i \in I$, $\phi_i$ is either (i)~concave and strictly increasing or (ii)~constant.

        \item \label{dom_mcs_charac:flip} Property~\ref{item:NSDflip} holds if and only if property~\ref{item:NWDflip} holds, which holds if and only if for each player $i \in I$, $\phi_i$ is convex and strictly increasing.
	
	\end{enumerate}
\end{corollary}

\section{Dominance by objects (rather than mixtures)}
\label{sec:pure}

Our results in the previous sections concerned the standard definition of `undominated', which deems an object undominated only if it is not dominated by any \emph{mixture} over objects. An alternative definition \parencite[in the spirit of][]{Borgers1993} worries only about dominance by other \emph{objects:}
\begin{align*}
    \widehat{\text{NSD}}(\objects,\situations,U) &\coloneqq
    \bigl\{
    x \in \objects :
    \text{there is no $y \in \objects$ such that}
    \\
    & \qquad \qquad \qquad
    \text{$U(y,s) > U(x,s)$ for every $s \in \situations$}
    \bigr\}
    \\
	\text{and} \quad
    \widehat{\text{NWD}}(\objects,\situations,U) &\coloneqq
    \bigl\{
    x \in \objects : \text{there is no $y \in \objects$ such that}
    \\
    {}& \qquad \qquad \qquad
    \text{$U(y,s) \geq U(x,s)$ for every $s \in \situations$ and}
    \\
    {}& \qquad \qquad \qquad
    \text{$U(y,s) > U(x,s)$ for some $s \in \situations$}
    \bigr\} .
\end{align*}

In this section, we study comparative-statics for these alternative notions of undominatedness. Concretely, we consider the following properties which a given payoff transformation $\Phi : \R^\Situations \to \R^\Situations$ may or may not satisfy:

\begin{enumerate}[label=($\hspace{0.1em}\widehat{\arabic*}\hspace{0.1em}$)]
\item \label{propertynpsd1} $\widehat{\text{NSD}}(\objects,\situations,U) \subseteq \widehat{\text{NSD}}(\objects,\situations,\Phi \circ U)$ for every finite problem $(\objects,\situations,U)$.
\item \label{propertynpsd2} $\widehat{\text{NWD}}(\objects,\situations,U) \subseteq \widehat{\text{NWD}}(\objects,\situations,\Phi \circ U)$ for every finite problem $(\objects,\situations,U)$.
\item \label{propertynpsd3} $\widehat{\text{NSD}}(\objects,\situations,\Phi \circ U) \subseteq \widehat{\text{NSD}}(\objects,\situations,U)$ for every finite problem $(\objects,\situations,U)$.
\item \label{propertynpsd4} $\widehat{\text{NWD}}(\objects,\situations,\Phi \circ U) \subseteq \widehat{\text{NWD}}(\objects,\situations,U)$ for every finite problem $(\objects,\situations,U)$.
\end{enumerate}
For each of these four properties, our question is which transformations $\Phi$ enjoy that property. The following theorem provides the answer.

\begin{theorem}\label{theorem:psd}
	Assume that $2 \leq \abs*{\Situations} < \infty$, and consider a transformation $\Phi : \R^\Situations \to \R^\Situations$.

	\begin{enumerate}[label=(\alph*)]
	
		\item \label{npsd1} Property~\ref{propertynpsd1} holds if and only if there exists a family $(\phi_s)_{s \in \Situations}$ of weakly increasing functions $\R \to \R$ such that $\Phi$ is $(\phi_s)_{s \in \Situations}$-case-wise.

		\item \label{npsd2} Property~\ref{propertynpsd2} holds if and only if there exists a family $(\phi_s)_{s \in \Situations}$ of functions $\R \to \R$ such that $\Phi$ is $(\phi_s)_{s \in \Situations}$-case-wise and either (i)~$\phi_s$ is strictly increasing for each $s \in \Situations$ or (ii)~$\phi_s$ is constant for each $s \in \Situations$.

        \item \label{npsd3} Property~\ref{propertynpsd3} holds if and only if for all $u,v \in \R^\Situations$ and $s \in \Situations$, $u_s > v_s$ implies $\Phi_s(u) > \Phi_s(v)$.
	    
        \item \label{npsd4} Property~\ref{propertynpsd4} holds if and only if for all $s,t \in \Situations$ and $u,v \in \R^\Situations$ such that $u_s > v_s$ and $u_t \geq v_t$, $\Phi_s(u) > \Phi_s(v)$ and $\Phi_t(u) \geq \Phi_t(v)$.
        
	\end{enumerate}
\end{theorem}

The first two parts are straightforward to show, and the latter two are immediate. Note that case-wiseness remains necessary for expansion, but not for contraction, as illustrated by the following example.

\begin{example}
\label{example:pure_notcasewise}
Let $\Situations=\{1,2\}$, and let $\Phi : \R^2 \to \R^2$ be given by
\begin{equation*}
    \Phi(k,\ell) \coloneqq
    \begin{cases}
    (k,\ell) & \text{if $k<0$ or if $k=0$ and $\ell\le 0$}\\
    (k+1,\ell) & \text{if $k>0$ or if $k=0$ and $\ell>0$.}
    \end{cases}
\end{equation*}
By \Cref{theorem:psd}, $\Phi$ satisfies properties~\ref{propertynpsd3} and \ref{propertynpsd4}. It is not case-wise.
\end{example}

In spite of \Cref{theorem:psd}, case-wiseness \emph{is} necessary for contraction when attention is restricted to moderately well-behaved transformations $\Phi$:

\begin{observation}\label{prop:finalprop}
Assume that $2 \leq \abs*{\Situations} < \infty$, and consider a \emph{continuous} transformation $\Phi : \R^\Situations \to \R^\Situations$. Assume that for all $u,v \in \R^\Situations$ and $s \in \Situations$, $u_s > v_s$ implies $\Phi_s(u) > \Phi_s(v)$. Then there exists a family $(\phi_s)_{s \in \Situations}$ of functions $\R \to \R$ such that $\Phi$ is $(\phi_s)_{s \in \Situations}$-case-wise and $\phi_s$ is strictly increasing for each $s \in \Situations$.
\end{observation}

\begin{proof}
Fix any $u,v \in \R^\Situations$ and $s \in \Situations$ such that $u_s = v_s$; we must show that $\Phi_s(u) = \Phi_s(v)$. To that end, for each $\varepsilon>0$, define $\underline{u}^\varepsilon,\overline{u}^\varepsilon \in \R^\Situations$ by $\underline{u}^\varepsilon_s \coloneqq u_s - \varepsilon$, $\overline{u}^\varepsilon_s \coloneqq u_s + \varepsilon$ and $\underline{u}^\varepsilon_t \coloneqq u_t \eqqcolon \overline{u}^\varepsilon_t$ for every $t \in \Situations \setminus \{s\}$. By the monotonicity assumption, $\Phi_s\left(\underline{u}^\varepsilon\right) < \Phi_s(v) < \Phi_s\left(\overline{u}^\varepsilon\right)$ for every $\varepsilon>0$. Since $\Phi$ is continuous, letting $\varepsilon \to 0$ yields $\Phi_s(u) = \Phi_s(v)$, as desired.
\end{proof}


\begin{appendices}

\crefalias{section}{appsec}
\crefalias{subsection}{appsec}
\crefalias{subsubsection}{appsec}

\section{Self-contained proof of the monotone specialisation of \texorpdfstring{\hyperref[theorem:sierpinski]{Sierpiński's theorem}}{Sierpiński's theorem}}
\label{pftheorem:sierpinski}

Fix a non-empty open interval $I \subseteq \R$, and let $f : I \to \R$ be weakly increasing. Obviously if $f$ is concave then
\begin{equation}
	\label{eq:cvx}
	\alpha f(u) + (1-\alpha) f(v)
    \leq f( \alpha u + (1-\alpha) v )
    \tag{$\lozenge$}
\end{equation}
holds for $\alpha=1/2$ and all $u<v$ in $I$. For the converse, assume that \eqref{eq:cvx} holds for $\alpha=1/2$ and all $u<v$ in $I$; we must show that $f$ is concave. It suffices to show that \eqref{eq:cvx} holds for all $u<v$ in $I$ and all dyadic-rational $\alpha \in [0,1]$, since then for any $\alpha \in [0,1]$, it holds for each dyadic-rational $\beta \in [\alpha,1]$ that
\begin{equation*}
\beta f(u) + (1-\beta) f(v)
\leq f( \beta u + (1-\beta) v )
\leq f( \alpha u + (1-\alpha) v ) 
\end{equation*}
since $f$ is weakly increasing, so letting $\beta \searrow \alpha$ (possible since the dyadic rationals are dense in $\R$) yields \eqref{eq:cvx}.

For each $n \in \N$, define
$$D_n \coloneqq \left\{ k 2^{-n} : k \in \left\{0,1,2,\dots,2^n-1,2^n\right\} \right\}.$$
We will show, for each $n \in \N$, that \eqref{eq:cvx} holds for all $u<v$ in $I$ and all $\alpha \in D_n$. We proceed by induction on $n \in \N$. The base case $n=1$ is immediate.

For the induction step, fix an arbitrary $n \in \N$, assume that \eqref{eq:cvx} holds for all $u<v$ in $I$ and all $\alpha \in D_n$, and fix arbitrary $u<v$ in $I$ and $\alpha \in D_{n+1}$; we will show that \eqref{eq:cvx} holds. Let $k \coloneqq \alpha 2^{n+1}$. If $k$ is even, then $\alpha \in D_n$, so \eqref{eq:cvx} holds by the induction hypothesis. Assume for the remainder that $k$ is odd.

Define $\underline{\alpha} \coloneqq (k-1)2^{-(n+1)}$ and $\overline{\alpha} \coloneqq (k+1)2^{-(n+1)}$. Note that $\underline{\alpha}$ and $\overline{\alpha}$ both belong to $D_n$ (since $k$ is odd), and that $\left(\underline{\alpha} + \overline{\alpha}\right) / 2 = \alpha$. We have 
\begin{align*}
\alpha f(u) + (1-\alpha) f(v)
&= \frac{1}{2} \left[ \underline{\alpha} f(u) + (1-\underline{\alpha}) f(v) \right]
+ \frac{1}{2} \left[ \overline{\alpha} f(u) + \left(1-\overline{\alpha}\right) f(v) \right]
\\
&\leq \frac{1}{2} f\left( \underline{\alpha} u + (1-\underline{\alpha}) v \right)
+ \frac{1}{2} f\left( \overline{\alpha} u + (1-\overline{\alpha}) v \right) 
\\
&\leq f\left( \frac{1}{2} \left[ \underline{\alpha} u + \left(1-\underline{\alpha}\right) v \right]
+ \frac{1}{2} \left[ \overline{\alpha} u + (1-\overline{\alpha}) v \right] \right)
\\
&= f( \alpha u + (1-\alpha) v ) ,
\end{align*}
where both inequalities follow from the induction hypothesis (the first since $\underline{\alpha} \in D_n \ni \overline{\alpha}$, and the second since $\frac{1}{2} \in D_1 \subseteq D_n$).
\qed

\section{Proof of sufficiency in \texorpdfstring{\Cref{proposition:dom_component}}{Proposition~\ref{proposition:dom_component}}}
\label{proposition:dom_component_suff}

Note that property~\ref{item:maintheorem1} implies property~\ref{item:NSDcomponent}, that property~\ref{item:maintheorem2} implies property~\ref{item:NWDcomponent}, and that property~(i) in \Cref{proposition:dom_component}\ref{dom_component:nsd_main} is the same as property~(i) in \Cref{proposition:dom_component}\ref{dom_component:nsd}. Hence to establish the sufficiency claims in \Cref{proposition:dom_component}, it is enough to show that
\begin{itemize}
\item property~(i) in \Cref{proposition:dom_component}\ref{dom_component:nsd_main}/\ref{dom_component:nsd} implies property~\ref{item:maintheorem1},
\item property~(i) in \Cref{proposition:dom_component}\ref{dom_component:nwd} implies property~\ref{item:maintheorem2},
\item property~(ii) in \Cref{proposition:dom_component}\ref{dom_component:nsd_main} implies property~\ref{item:maintheorem1},
\item property~(ii) in \Cref{proposition:dom_component}\ref{dom_component:nsd} implies property~\ref{item:NSDcomponent}, and
\item property~(ii) in \Cref{proposition:dom_component}\ref{dom_component:nwd} implies property~\ref{item:maintheorem2}.
\end{itemize}

Suppose that property~(i) in \Cref{proposition:dom_component}\ref{dom_component:nsd_main}/\ref{dom_component:nsd} holds: $\phi_s$ is concave and \emph{weakly} increasing for each $s \in \Situations$. We must show that property~\ref{item:maintheorem1} holds. To that end, fix a finite problem $(\objects,\situations,U)$ and any $x \in \objects \setminus \text{NSD}(\objects,\situations,\Phi \circ U)$; we will show that $x \notin \text{NSD}(\objects,\situations,U)$. By hypothesis, there is a mixture $p \in \Delta(\objects)$ such that
\begin{equation*}
    \sum_{y \in \objects} \phi_s(U(y,s)) p(y) > \phi_s(U(x,s)) \quad \text{for every $s \in \situations$.}
\end{equation*}
Hence for each $s \in \situations$, by Jensen's inequality (applicable since $\phi_s$ is concave),
\begin{equation*}
    \phi_s\left( \sum_{y \in \objects} U(y,s) p(y) \right)
    \geq \sum_{y \in \objects} \phi_s\left( U(y,s) \right) p(y)
    > \phi_s(U(x,s)) ,
\end{equation*}
whence $\sum_{y \in \objects} U(y,s) p(y) > U(x,s)$ since $\phi_s$ is weakly increasing.%
    \footnote{A function $f : \R \to \R$ is weakly increasing if and only if $u \leq v$ implies $f(u) \leq f(v)$, which is equivalent (by contra-position) to `$f(u) > f(v)$ implies $u>v$'.}
So $x \notin \text{NSD}(\objects,\situations,U)$.

Suppose that property~(i) in \Cref{proposition:dom_component}\ref{dom_component:nwd} holds: $\phi_s$ is concave and \emph{strictly} increasing for each $s \in \Situations$. We must show that property~\ref{item:maintheorem2} holds. To that end, fix a finite problem $(\objects,\situations,U)$ and any $x \in \objects \setminus \text{NWD}(\objects,\situations,\Phi \circ U)$; we will show that $x \notin \text{NWD}(\objects,\situations,U)$. By hypothesis, there is a mixture $p \in \Delta(\objects)$ such that
\begin{equation*}
    \sum_{y \in \objects} \phi_s(U(y,s)) p(y) \underset{(>)}{\geq} \phi_s(U(x,s)) \quad \text{for $\underset{\text{(some)}}{\text{every}}$ $s \in \situations$.}
\end{equation*}
For each $s \in \situations$, since $\phi_s$ is concave and strictly increasing,
\begin{equation*}
    \sum_{y \in \objects} \phi_s(U(y,s)) p(y)
    \underset{(>)}{\geq} \phi_s(U(x,s))
    \quad \text{implies} \quad
    \sum_{y \in \objects} U(y,s) p(y)
    \underset{(>)}{\geq} U(x,s) 
\end{equation*}
by \hyperref[theorem:pratt]{Pratt's theorem}. Hence $x \notin \text{NWD}(\objects,\situations,U)$.

Suppose that property~(ii) in \Cref{proposition:dom_component}\ref{dom_component:nsd_main} holds: there exists an $s \in \Situations$ such that $\phi_s$ is weakly increasing and for every $t \in \Situations \setminus \{s\}$, $\phi_t$ is constant. Then property~\ref{item:maintheorem1} holds since for any finite problem $(\objects,\situations,U)$, if $\situations \neq \{s\}$ then $\text{NSD}(\objects,\situations,U) \subseteq \objects = \text{NSD}(\objects,\situations,\Phi \circ U)$, while if $\situations = \{s\}$ then for any $x \in \text{NSD}(\objects,\situations,U)$, we have $U(x,s) \geq U(y,s)$ for every $y \in \objects$, which since $\phi_s$ is weakly increasing implies $\phi_s(U(x,s)) \geq \phi_s(U(y,s))$ for every $y \in \objects$, which is to say that $x \in \text{NSD}(\objects,\situations,\Phi \circ U)$.

Suppose that property~(ii) in \Cref{proposition:dom_component}\ref{dom_component:nsd} holds: $\phi_s$ is constant for some $s \in \Situations$. Then property~\ref{item:NSDcomponent} holds since for any finite problem $(\objects,\Situations,U)$, $\text{NSD}(\objects,\Situations,U) \subseteq \objects = \text{NSD}(\objects,\Situations,\Phi \circ U)$.

Suppose that property~(ii) in \Cref{proposition:dom_component}\ref{dom_component:nwd} holds: $\phi_s$ is constant for every $s \in \Situations$. Then property~\ref{item:maintheorem2} holds since for any finite problem $(\objects,\situations,U)$, $\text{NWD}(\objects,\situations,U) \subseteq \objects = \text{NWD}(\objects,\situations,\Phi \circ U)$.
\qed

\section{Proof of sufficiency in \texorpdfstring{\Cref{proposition:dom_component2}}{Proposition~\ref{proposition:dom_component2}}}
\label{proposition:dom_component2_suff}

Suppose that $\phi_s$ is convex and strictly increasing for each $s \in \Situations$; we must show that properties~\ref{item:shrink1} and \ref{item:shrink2} hold. This establishes the sufficiency claims in \Cref{proposition:dom_component2} since property~\ref{item:shrink1} implies property~\ref{item:NSDcomponentshrink} and property~\ref{item:shrink2} implies property~\ref{item:NWDcomponentshrink}.

To establish property~\ref{item:shrink2}, fix a finite problem $(\objects,\situations,U)$ and an $x \in \objects \setminus \text{NWD}(\objects,\situations,U)$; we will show that $x \notin \text{NWD}(\objects,\situations,\Phi \circ U)$. By hypothesis, there is a mixture $p \in \Delta(\objects)$ such that
\begin{equation*}
    \sum_{y \in \objects} U(y,s) p(y) \underset{(>)}{\geq} U(x,s) \quad \text{for $\underset{\text{(some)}}{\text{every}}$ $s \in \situations$.}
\end{equation*}
For each $s \in \situations$, since $\phi_s$ is convex and strictly increasing,
\begin{equation*}
    \sum_{y \in \objects} U(y,s) p(y)
    \underset{(>)}{\geq} U(x,s)
    \quad \text{implies} \quad
    \sum_{y \in \objects} \phi_s(U(y,s)) p(y)
    \underset{(>)}{\geq} \phi_s(U(x,s)) 
\end{equation*}
by \hyperref[theorem:pratt]{Pratt's theorem}. Hence $x \notin \text{NWD}(\objects,\situations,\Phi \circ U)$.

Property~\ref{item:shrink1} follows similarly: for any finite problem $(\objects,\situations,U)$, if $x \in \objects \setminus \text{NSD}(\objects,\situations,U)$, then there is a mixture $p \in \Delta(\objects)$ such that
\begin{equation*}
    \sum_{y \in \objects} U(y,s) p(y) > U(x,s) \quad \text{for every $s \in \situations$,}
\end{equation*}
which by \hyperref[theorem:pratt]{Pratt's theorem} implies
\begin{equation*}
    \sum_{y \in \objects} \phi_s(U(y,s)) p(y) > \phi_s(U(x,s)) \quad \text{for every $s \in \situations$,}
\end{equation*}
whence $x \notin \text{NSD}(\objects,\situations,\Phi \circ U)$.
\qed

\section{A lemma}
\label{proposition:dom_component_lemma}

The following purely technical lemma will be used in the next two appendices to prove the necessity claims in \Cref{proposition:dom_component,proposition:dom_component2}.

\begin{lemma}
	\label{lemma:cs}
    For any $\varepsilon \in (0,\infty)$, any non-empty open interval $I \subseteq \R$, and any weakly increasing and non-constant function $f : I \to \R$, there exist $u>v$ in $I$ and $\alpha \in [1/2,1/2+\varepsilon]$ such that $\alpha f(u) + (1-\alpha)f(v) > f\bigl(\frac{1}{2}u + \frac{1}{2}v\bigr)$.
\end{lemma}

\begin{proof}[Proof of \Cref{lemma:cs}]
	Suppose toward a contradiction that $f : I \to \R$ is weakly increasing and non-constant and that for some $\varepsilon \in (0,\infty)$, $\alpha f(u) + (1-\alpha)f(v) \leq f\bigl(\frac{1}{2}u + \frac{1}{2}v\bigr)$ for all $u>v$ in $I$ and $\alpha \in [1/2,1/2+\varepsilon]$. Since $f$ is Lebesgue-measurable (being weakly increasing), it follows by \hyperref[theorem:sierpinski]{Sierpiński's theorem} (\cref{sec:background:sierpinski}) that $f$ is concave. Since $f$ is concave, weakly increasing and non-constant, there must be a $w \in I$ such that $f'(w)$ exists and is strictly positive.%
        \footnote{Concavity implies absolute continuity, which by the (Lebesgue) fundamental theorem of calculus implies that $f'$ exists a.e. and $f(u)-f(v) = \int_v^u f'$ for all $u,v \in I$. We have $f' \geq 0$ since $f$ is weakly increasing, and we cannot have $f'=0$ since $f$ is non-constant.}
    Hence
    \begin{equation*}
		\lim_{\delta \to 0}
		\frac{\left(1/2+\varepsilon\right) f(w+\delta) + \left(1/2-\varepsilon\right)f(w-\delta) - f(w)}{\delta}
		= 2 \varepsilon f'(w) > 0 ,
	\end{equation*}
    so we may choose a (small) $\delta > 0$ such that
	\begin{equation*}
		\left(1/2+\varepsilon\right) f(w+\delta)
		+ \left(1/2-\varepsilon\right) f(w-\delta)
		> f(w) .
	\end{equation*}
	But then $\alpha f(u) + (1-\alpha)f(v) > f\bigl(\frac{1}{2}u + \frac{1}{2}v\bigr)$ holds for $\alpha \coloneqq 1/2+\varepsilon$, $u \coloneqq w+\delta$ and $v \coloneqq w-\delta$, a contradiction.
\end{proof}

\section{Proof of necessity in \texorpdfstring{\Cref{proposition:dom_component}}{Proposition~\ref{proposition:dom_component}}}
\label{proposition:dom_component_nec}

Note that property~\ref{item:maintheorem2} implies property~\ref{item:NWDcomponent}. Hence to establish the necessity claims in \Cref{proposition:dom_component}, it is enough to establish the following lemma.

\begin{lemma}
\label{prop1_nec}
Each of the following holds.
\begin{enumerate}[label=(\alph*)]
\item \label{prop1a_nec} If properties~(i) and (ii) in \Cref{proposition:dom_component}\ref{dom_component:nsd_main} both fail, then property~\ref{item:maintheorem1} fails.
\item \label{prop1b_nec} If properties~(i) and (ii) in \Cref{proposition:dom_component}\ref{dom_component:nsd} both fail, then property~\ref{item:NSDcomponent} fails.
\item \label{prop1c_nec} If properties~(i) and (ii) in \Cref{proposition:dom_component}\ref{dom_component:nwd} both fail, then property~\ref{item:NWDcomponent} fails.
\end{enumerate}
\end{lemma}

The proof of each part of \Cref{prop1_nec} is a variation on the core necessity argument laid out in \cref{sec:result:props} (in the proof of \Cref{corr:dom_component_single}), and uses \Cref{lemma:cs} from \cref{proposition:dom_component_lemma}.

\begin{proof}[Proof of \Cref{prop1_nec}\ref{prop1a_nec}] Fix arbitrary $x \neq y \neq z \neq x$. Assume that properties~(i) and (ii) in \Cref{proposition:dom_component}\ref{dom_component:nsd_main} both fail; we must show that property~\ref{item:maintheorem1} fails. Note that property~(i) can fail in two ways: either $\phi_s$ fails to be weakly increasing for some $s \in \Situations$, or $\phi_s$ fails to be concave for some $s \in \Situations$. Similarly, property~(ii) can fail in two ways: either there is an $s \in \Situations$ such that $\phi_s$ is not weakly increasing, or $\phi_s$ is non-constant for at least two distinct $s \in \Situations$.

Suppose that property~(i) fails in the first way or that property~(ii) fails in the first way; in either scenario, there is an $s \in \Situations$ such that $\phi_s$ is not weakly increasing, meaning that there are $u>v$ such that $\phi_s(u)<\phi_s(v)$. Then for $U : \{x,y\} \times \Situations \to \R$ satisfying $U(x,s) = u$ and $U(y,s) = v$, it holds that $\text{NSD}(\{x,y\},\{s\},U) \ni x \notin \text{NSD}(\{x,y\},\{s\},\Phi \circ U)$, so property~\ref{item:maintheorem1} fails.

Assume for the remainder that properties~(i) and (ii) do not fail in the first way. Then they both fail in the second way: there are $s \neq t$ in $\Situations$ such that $\phi_s$ is weakly increasing but not concave and $\phi_t$ is weakly increasing and non-constant. By \hyperref[theorem:sierpinski]{Sierpiński's theorem} from \cref{sec:background:sierpinski} above (applicable since increasing functions are Lebesgue-measurable), there exist $u>v$ such that $\frac{1}{2} \phi_s(u) + \frac{1}{2} \phi_s(v) > \phi_s\bigl( \frac{1}{2} u + \frac{1}{2} v \bigr)$. Then there is a (small) $\varepsilon \in (0,\infty)$ such that $(1-\beta) \phi_s(u) + \beta \phi_s(v) > \phi_s\bigl( \frac{1}{2} u + \frac{1}{2} v \bigr)$ for every $\beta \in [1/2,1/2+\varepsilon]$. By \Cref{lemma:cs} from \cref{proposition:dom_component_lemma} (applicable since $\phi_t$ is weakly increasing and non-constant), there exist $u'>v'$ and $\alpha \in [1/2,1/2+\varepsilon]$ such that $\alpha \phi_t(u') + (1-\alpha) \phi_t(v') > \phi_t\bigl( \frac{1}{2} u' + \frac{1}{2} v' \bigr)$. Then for $U : \{x,y,z\} \times \Situations \to \R$ satisfying $U(x,s) = u$, $U(x,t) = v'$, $U(y,s) = v$, $U(y,t) = u'$, $U(z,s) = \frac{1}{2} u + \frac{1}{2} v$ and $U(z,t) = \frac{1}{2} u' + \frac{1}{2} v'$, we have $z \in \text{NSD}(\{x,y,z\},\{s,t\},U)$, and keeping in mind the mixture $p \in \Delta(\{x,y,z\})$ with $p(x) \coloneqq 1-\alpha$ and $p(y) \coloneqq \alpha$, we see that $z \notin \text{NSD}(\{x,y,z\},\{s,t\},\Phi \circ U)$. Thus property~\ref{item:maintheorem1} fails.
\end{proof}

\begin{proof}[Proof of \Cref{prop1_nec}\ref{prop1b_nec}]
Fix arbitrary $x \neq y \neq z \neq x$. Assume that properties~(i) and (ii) in \Cref{proposition:dom_component}\ref{dom_component:nsd} both fail; we must show that property~\ref{item:NSDcomponent} fails. The fact that property~(ii) fails means that for every $s \in \Situations$, there are $u_s \neq v_s$ such that $\phi_s(u_s) > \phi_s(v_s)$. Property~(i) can fail in two ways: either $\phi_s$ fails to be weakly increasing for some $s \in \Situations$, or $\phi_s$ fails to be concave for some $s \in \Situations$.

Suppose that property~(i) fails in the first way: there is an $s \in \Situations$ such that $\phi_s$ is not weakly increasing, meaning that there exist $u>v$ such that $\phi_s(u)<\phi_s(v)$. Then for $U : \{x,y\} \times \Situations \to \R$ given by $U(x,s) \coloneqq u$, $U(y,s) \coloneqq v$, and $U(x,t) \coloneqq v_t$ and $U(y,t) \coloneqq u_t$ for every $t \in \Situations \setminus \{s\}$, it holds that $\text{NSD}(\{x,y\},\Situations,U) \ni x \notin \text{NSD}(\{x,y\},\Situations,\Phi \circ U)$, so property~\ref{item:NSDcomponent} fails.

Assume for the remainder that property~(i) does not fail in the first way. Then it fails in the second way: there is an $s \in \Situations$ such that $\phi_s$ is weakly increasing but not concave. By \hyperref[theorem:sierpinski]{Sierpiński's theorem} from \cref{sec:background:sierpinski} above (applicable since increasing functions are Lebesgue-measurable), there exist $u>v$ such that $\frac{1}{2} \phi_s(u) + \frac{1}{2} \phi_s(v) > \phi_s\bigl( \frac{1}{2} u + \frac{1}{2} v \bigr)$. Then there is a (small) $\varepsilon \in (0,\infty)$ such that $(1-\beta) \phi_s(u) + \beta \phi_s(v) > \phi_s\bigl( \frac{1}{2} u + \frac{1}{2} v \bigr)$ for every $\beta \in [1/2,1/2+\varepsilon]$. Fix an arbitrary $t \in \Situations \setminus \{s\}$. By \Cref{lemma:cs} from \cref{proposition:dom_component_lemma} (applicable since $\phi_t$ is weakly increasing and non-constant, as property~(ii) fails and property~(i) does not fail in the first way), there exist $u'>v'$ and $\alpha \in [1/2,1/2+\varepsilon]$ such that $\alpha \phi_t(u') + (1-\alpha) \phi_t(v') > \phi_t\bigl( \frac{1}{2} u' + \frac{1}{2} v' \bigr)$. Then for $U : \{x,y,z\} \times \Situations \to \R$ given by $U(x,s) \coloneqq u$, $U(y,s) \coloneqq v$, $U(z,s) \coloneqq \frac{1}{2} u + \frac{1}{2} v$, $U(x,t) \coloneqq v'$, $U(y,t) \coloneqq u'$, $U(z,t) \coloneqq \frac{1}{2} u' + \frac{1}{2} v'$, and $U(x,r) \coloneqq u_r$, $U(y,r) \coloneqq u_r$ and $U(z,r) \coloneqq v_r$ for every $r \in \Situations \setminus \{s,t\}$, we have $z \in \text{NSD}(\{x,y,z\},\Situations,U)$, and keeping in mind the mixture $p \in \Delta(\{x,y,z\})$ with $p(x) \coloneqq 1-\alpha$ and $p(y) \coloneqq \alpha$, we see that $z \notin \text{NSD}(\{x,y,z\},\Situations,\Phi \circ U)$. Thus property~\ref{item:NSDcomponent} fails.
\end{proof}

\begin{proof}[Proof of \Cref{prop1_nec}\ref{prop1c_nec}]
Fix arbitrary $x \neq y \neq z \neq x$. Assume that properties~(i) and (ii) in \Cref{proposition:dom_component}\ref{dom_component:nwd} both fail; we must show that property~\ref{item:NWDcomponent} fails. The fact that property~(ii) fails means that there exists a $t \in \Situations$ such that $\phi_t$ is non-constant, i.e. $\phi_t(u_t) > \phi_t(v_t)$ for some $u_t \neq v_t$. Property~(i) can fail in two ways: either $\phi_s$ fails to be strictly increasing for some $s \in \Situations$, or $\phi_s$ fails to be concave for some $s \in \Situations$.

Suppose that property~(i) fails in the first way: there is an $s \in \Situations$ such that $\phi_s$ is not strictly increasing, meaning that there are $u>v$ such that $\phi_s(u) \leq \phi_s(v)$. Since $\Situations \ge 2$, we can choose $s$ and $t$ to be distinct. Then for $U : \{x,y\} \times \Situations \to \R$ given by $U(x,s) \coloneqq u$, $U(x,t) \coloneqq v_t$, $U(y,s) \coloneqq v$, $U(y,t) \coloneqq u_t$, and $U(x,r) \coloneqq 0 \eqqcolon U(y,r)$ for every $r \in \Situations \setminus \{s,t\}$, it holds that $\text{NWD}(\{x,y\},\Situations,U) \ni x \notin \text{NWD}(\{x,y\},\Situations,\Phi \circ U)$, so property~\ref{item:NWDcomponent} fails.

Assume for the remainder that property~(i) does not fail in the first way. Then property~(i) fails in the second way: there is an $s \in \Situations$ such that $\phi_s$ is strictly increasing but not concave. Since $\Situations \ge 2$, we can once again choose $s$ and $t$ to be distinct.
 By \hyperref[theorem:sierpinski]{Sierpiński's theorem} from \cref{sec:background:sierpinski} above (applicable since increasing functions are Lebesgue-measurable), there exist $u>v$ such that $\frac{1}{2} \phi_s(u) + \frac{1}{2} \phi_s(v) > \phi_s\bigl( \frac{1}{2} u + \frac{1}{2} v \bigr)$. Then there is a (small) $\varepsilon \in (0,\infty)$ such that $(1-\beta) \phi_s(u) + \beta \phi_s(v) > \phi_s\bigl( \frac{1}{2} u + \frac{1}{2} v \bigr)$ for every $\beta \in [1/2,1/2+\varepsilon]$. By \Cref{lemma:cs} from \cref{proposition:dom_component_lemma} (applicable since $\phi_t$ is strictly increasing as property~(i) does not fail in the first way), there exist $u'>v'$ and $\alpha \in [1/2,1/2+\varepsilon]$ such that $\alpha \phi_t(u') + (1-\alpha) \phi_t(v') > \phi_t\bigl( \frac{1}{2} u' + \frac{1}{2} v' \bigr)$. Then for $U : \{x,y,z\} \times \Situations \to \R$ given by $U(x,s) \coloneqq u$, $U(y,s) \coloneqq v$, $U(z,s) \coloneqq \frac{1}{2} u + \frac{1}{2} v$, $U(x,t) \coloneqq v'$, $U(y,t) \coloneqq u'$, $U(z,t) \coloneqq \frac{1}{2} u' + \frac{1}{2} v'$, and $U(x,r) \coloneqq 0$, $U(y,r) \coloneqq 0$ and $U(z,r) \coloneqq 0$ for every $r \in \Situations \setminus \{s,t\}$, we have $z \in \text{NWD}(\{x,y,z\},\Situations,U)$, and keeping in mind the mixture $p \in \Delta(\{x,y,z\})$ with $p(x) \coloneqq 1-\alpha$ and $p(y) \coloneqq \alpha$, we see that $z \notin \text{NWD}(\{x,y,z\},\Situations,\Phi \circ U)$. Thus property~\ref{item:NWDcomponent} fails.
\end{proof}

\section{Proof of necessity in \texorpdfstring{\Cref{proposition:dom_component2}}{Proposition~\ref{proposition:dom_component2}}}
\label{proposition:dom_component2_nec}

Given a strictly increasing $\phi : \R \to \R$, let $\phi^{-1} : \co(\phi(\R)) \to \R$ be given by
\begin{equation*}
	\phi^{-1}(w) \coloneqq \inf\{w' \in \R : \phi(w') \ge w\}
	\quad \text{for each $w \in \co(\phi(\R))$,}
\end{equation*} 
and let $D(\phi)$ be the set of $(u,v) \in \co(\phi(\R))^2$ such that $\phi(\phi^{-1}(u)) \le u$, $\phi(\phi^{-1}(v)) \le v$, and $\phi\bigl(\phi^{-1}\bigl(\frac{1}{2} u + \frac{1}{2} v\bigr)\bigr) \ge \frac{1}{2} u + \frac{1}{2} v$.

\begin{lemma}
	\label{lemma:gen_inv}
	Let $\phi : \R \to \R$ be strictly increasing and not continuous. Then there exists $(u,v) \in D(\phi)$ such that $\frac{1}{2} \phi^{-1}(u) + \frac{1}{2} \phi^{-1}(v) > \phi^{-1}\bigl( \frac{1}{2} u + \frac{1}{2} v \bigr)$.
\end{lemma}

\begin{proof}
	Observe that $\phi^{-1}$ is continuous since $\phi$ is strictly increasing. Fix a $w \in \R$ at which $\phi$ is not continuous, set
	\begin{equation*}
		\phi(w-) \coloneqq \sup_{w' \in (-\infty,w)} \phi(w')
		\quad \text{and} \quad
		\phi(w+) \coloneqq \inf_{w' \in (w,\infty)} \phi(w') ,
	\end{equation*}
	and consider two scenarios.

	\smallskip 
	\noindent \emph{Scenario~1: $\phi(w-)+\phi(w+) < 2\phi(w)$.} In this scenario, we may choose $v \in \phi(\R) \cap (\phi(w+),2\phi(w)-\phi(w-))$. 
	Since $\phi$ is strictly increasing and $v > \phi(w+)$, we have $\phi^{-1}(v) > \phi^{-1}(\phi(w+)) = w$.
	Then, since $\phi^{-1}$ is continuous, choosing $u \in \phi(\R) \cap (2\phi(w-)-v,\phi(w-))$ such that $\phi^{-1}(u)$ is sufficiently close to $\phi^{-1}(\phi(w-)) = w$ yields $\frac{1}{2} \phi^{-1}(u) + \frac{1}{2} \phi^{-1}(v) > w$.
	Note that $\phi(w-) < \frac{1}{2} u + \frac{1}{2} v < \phi(w)$, so that $\phi^{-1}\bigl( \frac{1}{2} u + \frac{1}{2} v \bigr) = w $.
	Then $\frac{1}{2} \phi^{-1}(u) + \frac{1}{2} \phi^{-1}(v) > \phi^{-1}\bigl( \frac{1}{2} u + \frac{1}{2} v \bigr)$ and $\phi\bigl(\phi^{-1}\bigl(\frac{1}{2} u + \frac{1}{2} v\bigr)\bigr) = \phi(w) > \frac{1}{2} u + \frac{1}{2} v$. Hence $(u,v) \in D(\phi)$.

	\smallskip 
	\noindent 
	\emph{Scenario~2: $\phi(w-)+\phi(w+) \ge 2\phi(w)$.} Since $\phi(w-) < \phi(w+) > \phi(w)$ in this scenario, we may choose $v \in \phi(\R) \cap (\phi(w+),2\phi(w+)-\phi(w))$.
	Since $\phi$ is strictly increasing and $v > \phi(w+)$, we have $\phi^{-1}(v) > \phi^{-1}(\phi(w+))$. Then, since $\phi^{-1}$ is continuous, choosing $m \in \phi(\R) \cap (\phi(w+),\frac 12(\phi(w+)+v))$ such that $\phi^{-1}(m)$ is sufficiently close to $\phi^{-1}(\phi(w+))$ yields $\frac 12 \phi^{-1}(\phi(w+))+\frac 12 \phi^{-1}(v) > \phi^{-1}(m)$.
	Let $u \coloneqq 2m-v$. Then $\phi(w) < u < \phi(w+)$, so that $\phi^{-1}(u) = w = \phi^{-1}(\phi(w+))$, whence $\frac{1}{2} \phi^{-1}(u) + \frac{1}{2} \phi^{-1}(v) > \phi^{-1}(m) = \phi^{-1}\bigl( \frac{1}{2} u + \frac{1}{2} v \bigr)$ and $\phi(\phi^{-1}(u)) = \phi(w) < u$.
	Finally, $\phi(\phi^{-1}(v)) = v$ and $\phi(\phi^{-1}(m)) = m$ since $v \in \phi(\R) \ni m$, so that $(u,v) \in D(\phi)$.
\end{proof}

\begin{proof}[Proof of \Cref{proposition:dom_component2}]
	Note that property~\ref{item:shrink1} implies property~\ref{item:NSDcomponentshrink} and that property~\ref{item:shrink2} implies property~\ref{item:NWDcomponentshrink}. Hence to establish the necessity claims in \Cref{proposition:dom_component2}, it is enough to show that if 
	\begin{enumerate}[label=($\star$)]
	\item \label{starprop} $\phi_s$ is convex and strictly increasing for each $s \in \Situations$
	\end{enumerate}
	fails to hold, then properties~\ref{item:NSDcomponentshrink} and \ref{item:NWDcomponentshrink} both fail. The argument is a variation on the core necessity argument laid out in \cref{sec:result:props} (in the proof of \Cref{corr:dom_component_single}), and uses \Cref{lemma:cs} from \cref{proposition:dom_component_lemma}.

	Fix arbitrary $x \neq y \neq z \neq x$. Assume that property~\ref{starprop} fails; we must show that properties~\ref{item:NSDcomponentshrink} and \ref{item:NWDcomponentshrink} both fail. There are two ways in which property~\ref{starprop} can fail: either $\phi_s$ fails to be strictly increasing for some $s \in \Situations$, or $\phi_s$ fails to be convex for some $s \in \Situations$.

	Suppose that property~\ref{starprop} fails in the first way: there is an $s \in \Situations$ such that $\phi_s$ is not strictly increasing, meaning that there are $u>v$ such that $\phi_s(u) \leq \phi_s(v)$. Then for $U : \{x,y\} \times \Situations \to \R$ given by $U(x,s) \coloneqq v$, $U(x,t) \coloneqq 0$, $U(y,s) \coloneqq u$ and $U(y,t) \coloneqq 1$ for every $t \in \Situations \setminus \{s\}$, it holds that $\text{NSD}(\{x,y\},\Situations,U) \not\ni x \in \text{NSD}(\{x,y\},\Situations,\Phi \circ U)$, so property~\ref{item:NSDcomponentshrink} fails. Furthermore, for $V : \{x,y\} \times \Situations \to \R$ given by $V(\cdot,s) \coloneqq U(\cdot,s)$ and $V(x,t) \coloneqq 0$ and $V(y,t) \coloneqq 0$ for every $t \in \Situations \setminus \{s\}$, it holds that $\text{NWD}(\{x,y\},\Situations,V) \not\ni x \in \text{NWD}(\{x,y\},\Situations,\Phi \circ V)$, so property~\ref{item:NWDcomponentshrink} fails.

	Assume for the remainder that property~\ref{starprop} does not fail in the first way. Then it fails in the second way: there is an $s \in \Situations$ such that $\phi_s$ is strictly increasing but not convex. Then $\phi_s^{-1}$ is well-defined, increasing, and not concave. 
	Hence there exists $(u,v) \in D(\phi_s)$ such that $\frac{1}{2} \phi^{-1}_s(u) + \frac{1}{2} \phi^{-1}_s(v) > \phi^{-1}_s\bigl( \frac{1}{2} u + \frac{1}{2} v \bigr)$: if $\phi_s$ is continuous then this follows from \hyperref[theorem:sierpinski]{Sierpiński's theorem} from \cref{sec:background:sierpinski} above (applicable since increasing functions are Lebesgue-measurable and $\co(\phi_s(\R))$ is a non-empty open interval), while if $\phi_s$ is not continuous then it follows from \Cref{lemma:gen_inv}.
	Then there is a (small) $\varepsilon \in (0,\infty)$ such that $(1-\beta) \phi_s^{-1}(u) + \beta \phi_s^{-1}(v) > \phi_s^{-1}\bigl( \frac{1}{2} u + \frac{1}{2} v \bigr)$ for every $\beta \in [1/2,1/2+\varepsilon]$. Fix an arbitrary $t \in \Situations \setminus \{s\}$. Since property~\ref{starprop} does not fail in the first way, $\phi_t$ is strictly increasing, so $\phi_t^{-1}$ is well-defined, increasing, and non-constant. Then there exists $(u',v') \in D(\phi_t)$ and $\alpha \in [1/2,1/2+\varepsilon]$ such that $\alpha \phi_t^{-1}(u') + (1-\alpha) \phi_t^{-1}(v') > \phi_t^{-1}\bigl( \frac{1}{2} u' + \frac{1}{2} v' \bigr)$: if $\phi_t$ is continuous then this follows from \Cref{lemma:cs} in \cref{proposition:dom_component_lemma}, while if $\phi_t$ is not continuous then it follows from \Cref{lemma:gen_inv}. Then for $U : \{x,y,z\} \times \Situations \to \R$ given by $U(x,s) \coloneqq \phi_s^{-1}(u)$, $U(y,s) \coloneqq \phi_s^{-1}(v)$, $U(z,s) \coloneqq \phi_s^{-1}\bigl(\frac{1}{2} u + \frac{1}{2} v\bigr)$, $U(x,t) \coloneqq \phi_t^{-1}(v')$, $U(y,t) \coloneqq \phi_t^{-1}(u')$, $U(z,t) \coloneqq \phi_t^{-1}\bigl(\frac{1}{2} u' + \frac{1}{2} v'\bigr)$, and $U(x,r) \coloneqq 1$, $U(y,r) \coloneqq 1$ and $U(z,r) \coloneqq 0$ for every $r \in \Situations \setminus \{s,t\}$, we have $z \in \text{NSD}(\{x,y,z\},\Situations,\Phi \circ U)$, and keeping in mind the mixture $p \in \Delta(\{x,y,z\})$ with $p(x) \coloneqq 1-\alpha$ and $p(y) \coloneqq \alpha$, we see that $z \notin \text{NSD}(\{x,y,z\},\Situations,U)$. Thus property~\ref{item:NSDcomponentshrink} fails. Furthermore, for $V : \{x,y,z\} \times \Situations \to \R$ given by $V(\cdot,s) \coloneqq U(\cdot,s)$, $V(\cdot,t) \coloneqq U(\cdot,t)$, and $V(x,r) \coloneqq 0$, $V(y,r) \coloneqq 0$ and $V(z,r) \coloneqq 0$ for every $r \in \Situations \setminus \{s,t\}$, we have $z \in \text{NWD}(\{x,y,z\},\Situations,\Phi \circ V)$, and keeping in mind the mixture $p \in \Delta(\{x,y,z\})$ with $p(x) \coloneqq 1-\alpha$ and $p(y) \coloneqq \alpha$, we see that $z \notin \text{NWD}(\{x,y,z\},\Situations,V)$. Thus property~\ref{item:NWDcomponentshrink} fails.
\end{proof}

\section{Proof of \texorpdfstring{\Cref{thm:maintheorem}}{Theorem~\ref{thm:maintheorem}}}
\label{pf:maintheorem}

The following lemma asserts that case-wiseness is necessary.

\begin{lemma}\label{lem:sitwise_nec}
    Assume that $2 \leq \abs*{\Situations} < \infty$, and consider a transformation $\Phi : \R^\Situations \to \R^\Situations$.

    \begin{enumerate}[label=(\roman*)]
        
        \item \label{lem:case-wisedeux} If $\Phi$ is not case-wise, then properties~\ref{item:maintheorem1} and \ref{item:maintheorem2} both fail.

        \item \label{lem:case-wise} If $\Phi$ is not case-wise, then properties~\ref{item:shrink1} and \ref{item:shrink2} both fail.
        
    \end{enumerate}
\end{lemma}

\begin{proof}[Proof of \Cref{thm:maintheorem}]
    Sufficiency follows from \Cref{proposition:dom_component,proposition:dom_component2}. For necessity, in light of \Cref{proposition:dom_component,proposition:dom_component2}, all that must be shown is that case-wiseness is necessary for each of properties~\ref{item:maintheorem1}--\ref{item:shrink2}, and this follows from \Cref{lem:sitwise_nec}.
\end{proof}

As mentioned in \cref{sec:result}, \Cref{lem:sitwise_nec}\ref{lem:case-wisedeux} is almost immediate, whereas \Cref{lem:sitwise_nec}\ref{lem:case-wise} requires a non-trivial proof.

\begin{proof}[Proof of \Cref{lem:sitwise_nec}\ref{lem:case-wisedeux}]
Fix arbitrary $x \neq y$. Suppose that $\Phi$ is not case-wise; we will show that properties~\ref{item:maintheorem1} and \ref{item:maintheorem2} both fail. Since $\Phi$ is not case-wise, there exist $s \in \Situations$ and $u \neq v$ with $u_s=v_s$ such that $\Phi_s(u) < \Phi_s(v)$. Then for $U : \{x,y\} \times \Situations \to \R$ given by $U(x,\cdot) \coloneqq u$ and $U(y,\cdot) \coloneqq v$, we have $x \in \text{NWD}(\{x,y\},\{s\},U) \subseteq \text{NSD}(\{x,y\},\{s\},U)$ and $x \notin \text{NSD}(\{x,y\},\{s\},\Phi \circ U) \supseteq \text{NWD}(\{x,y\},\{s\},\Phi \circ U)$, so properties~\ref{item:maintheorem1} and \ref{item:maintheorem2} both fail.
\end{proof}

\begin{proof}[Proof of \Cref{lem:sitwise_nec}\ref{lem:case-wise}]
    Fix arbitrary $x \neq y \neq z \neq x$. Assume that either property~\ref{item:shrink1} or property~\ref{item:shrink2} holds; we will show that $\Phi$ is case-wise.
    
    We claim that $\Phi$ is \emph{strictly case-wise monotone} in the sense that for any $u,v \in \R^\Situations$ and any case $s \in \Situations$, $u_s > v_s$ implies $\Phi_s(u) > \Phi_s(v)$. To see why, suppose toward a contradiction that $u_s > v_s$ and $\Phi_s(u) \leq \Phi_s(v)$ for some $u,v \in \R^\Situations$ and $s \in \Situations$. Then letting $U : \{x,y\} \times \Situations \to \R$ be given by $U(x,\cdot) \coloneqq v$ and $U(y,\cdot) \coloneqq u$, we have $x \notin \text{NSD}(\{x,y\},\{s\},U) \supseteq \text{NWD}(\{x,y\},\{s\},U)$ but $x \in \text{NWD}(\{x,y\},\{s\},\Phi \circ U) \subseteq \text{NSD}(\{x,y\},\{s\},\Phi \circ U)$, so properties~\ref{item:shrink1} and \ref{item:shrink2} both fail---a contradiction.

    For each $s \in \Situations$, define $\phi_s : \R \to \R$ by $\phi_s(u_s) \coloneqq \Phi_s(u_s \1^s)$ for every $u_s \in \R$, where $\1^s \in \{0,1\}^\Situations$ is given by $\1^s_s \coloneqq 1$ and $\1^s_t \coloneqq 0$ for every $t \in \Situations \setminus \{s\}$. Strict case-wise monotonicity implies that for each $s \in \Situations$, $\phi_s$ is strictly increasing and satisfies $\phi_s(u_s-\varepsilon) < \Phi_s(u) < \phi_s(u_s+\varepsilon)$ for every $u \in \R^\Situations$ and $\varepsilon > 0$. Letting $\varepsilon$ vanish yields
	\begin{equation}
		\label{eq:sandwich}
		\phi_s(u_s-) \le \Phi_s(u) \le \phi_s(u_s+) \quad \text{for all $u \in \R^\Situations$.}
        \tag{$\ast$}
	\end{equation}
    Hence to show that $\Phi$ is $(\phi_s)_{s \in \Situations}$-case-wise, it suffices to show that $\phi_s$ is continuous for each $s \in \Situations$.

    To that end, suppose toward a contradiction that there are $s \in \Situations$ and $u \in \R$ such that $\phi_s(u-) < \phi_s(u+)$. It suffices to exhibit $a,b,c \in \R^\Situations$ such that $\text{NSD}(\{a,b,c\},\Situations,\mathrm{id}) \not\ni c \in \text{NWD}(\{a,b,c\},\Situations,\Phi)$, where $\mathrm{id}(d,\cdot) \coloneqq d$ for each $d \in \{a,b,c\}$, since then properties~\ref{item:shrink1} and \ref{item:shrink2} both fail---a contradiction.

    Fix an arbitrary $t \in \Situations \setminus \{s\}$. Since $\phi_s(u-) < \phi_s(u+)$ and since $\phi_s$ and $\phi_t$ are continuous a.e. (being increasing), we may choose a $v \in (u,\infty)$ such that
	\begin{equation}
		\label{eq:v}
		\phi_s(u+) > \frac{\phi_s(u-)+\phi_s(v)}2
        \tag{$\dagger$}
	\end{equation}
    and $\phi_s$ and $\phi_t$ are both continuous at $v$. Since $\phi_t$ is continuous at $v$ and strictly increasing, there exists a $w \in (u,v)$ such that 
	\begin{equation}
		\label{eq:w}
		\phi_t(w-) > \frac{\phi_t(u+) + \phi_t(v)}{2} .
        \tag{$\ddagger$}
	\end{equation}
    Choose $v' \in (-\infty,u)$ and $w' \in (u,w)$ close enough to $u$ that $\underline{\alpha} \coloneqq (w-v')/(v-v')$ and $\overline{\alpha} \coloneqq (v-w')/(v-v')$ satisfy $\underline{\alpha} < \overline{\alpha}$, and such that both $\phi_s$ and $\phi_t$ are continuous at $v'$. Note that $v'<u<w'<w<v$. Define $a,b,c \in \R^\Situations$ by $a_s \coloneqq v$, $b_s \coloneqq v'$, $c_s \coloneqq w'$, and $a_r \coloneqq v'$, $b_r \coloneqq v$, $c_r \coloneqq w$ for every $r \in \Situations \setminus \{s\}$. Then, keeping in mind a mixture $p \in \Delta(\{a,b,c\})$ with $\underline{\alpha} < p(b) < \overline{\alpha}$ and $p(a) = 1-p(b)$, we see that $c \notin \text{NSD}(\{a,b,c\},\Situations,\mathrm{id})$.
	
	It remains to show that $c \in \text{NWD}(\{a,b,c\},\Situations,\Phi)$. To this end, note that $b_s < c_s < a_s$ and $a_t < c_t < b_t$, so that $\Phi_s(b) < \Phi_s(c) < \Phi_s(a)$ and $\Phi_t(a) < \Phi_t(c) < \Phi_t(b)$ by strict case-wise monotonicity. It therefore suffices to show that 
	\begin{equation*}
		\Phi_r(c) > \frac {\Phi_r(a)+\Phi_r(b)} 2 
        \quad \text{for $r=s$ and for $r=t$.}\footnotemark
	\end{equation*}
    \footnotetext{For any mixture $p \in \Delta(\{a,b,c\})$ with $p(c)<1$, if $p(a) \leq p(b)$ then $\Phi_s(c) > [1-p(c)] [ \Phi_s(a) + \Phi_s(b) ] / 2 + p(c) \Phi_s(c) \geq p(a) \Phi_s(a) + p(b) \Phi_s(b) + p(c) \Phi_s(c)$, and if $p(a) \geq p(b)$ then $\Phi_t(c) > [1-p(c)] [ \Phi_t(a) + \Phi_t(b) ] / 2 + p(c) \Phi_t(c) \geq p(a) \Phi_t(a) + p(b) \Phi_t(b) + p(c) \Phi_t(c)$.}
	For the scenario $r = s$, we have
	\begin{multline*}
		\Phi_s(c)
        \ge \phi_s(w'-)
        \ge \phi_s(u+)
        \\
        > \frac {\phi_s(u-)+\phi_s(v)}2
        \ge \frac{\phi_s(v')+\phi_s(v)}{2} 
		= \frac{\Phi_s(b)+\Phi_s(a)}{2},
	\end{multline*}
	where the first inequality holds by \eqref{eq:sandwich}, the second inequality holds since $w' > u$ and $\phi_s$ is increasing, the third inequality is \eqref{eq:v}, the fourth inequality holds since $u > v'$ and $\phi_s$ is increasing, and the equality follows from \eqref{eq:sandwich} and the fact that $\phi_s$ is continuous at $v$ and at $v'$.
	For the scenario $r = t$, we have
	\begin{equation*}
		\Phi_t(c) \ge \phi_t(w-) > \frac{\phi_t(u+) + \phi_t(v)}{2} \ge \frac{\phi_t(v') + \phi_t(v)}{2} = \frac {\Phi_t(a) + \Phi_t(b)}2
	\end{equation*}
	where the first inequality holds by \eqref{eq:sandwich}, the second inequality is \eqref{eq:w}, the third inequality holds since $u > v'$ and $\phi_t$ is increasing, and the equality follows from \eqref{eq:sandwich} and the fact that $\phi_t$ is continuous at $v$ and at $v'$.
\end{proof}

\section{Proof of \texorpdfstring{\Cref{theorem:psd}}{Theorem~\ref{theorem:psd}}}

The sufficiency (`if') claims in \Cref{theorem:psd}\ref{npsd1}--\ref{npsd4} are easily established. It remains to show necessity (the `only if' claims). Fix arbitrary $x \neq y$.

For the necessity (`only if') claims in \Cref{theorem:psd}\ref{npsd1}--\ref{npsd2}, assume that either property~\ref{propertynpsd1} or property~\ref{propertynpsd2} holds. Mimicking the proof of \Cref{lem:sitwise_nec}\ref{lem:case-wisedeux} in \cref{pf:maintheorem} yields that $\Phi$ must be $(\phi_s)_{s \in \Situations}$-case-wise for some family $(\phi_s)_{s \in \Situations}$ of functions $\R \to \R$. Furthermore, $\phi_s$ must be weakly increasing for every $s \in \Situations$. To see why, suppose toward a contradiction that there are $s \in \Situations$ and $u_s > v_s$ in $\R$ such that $\phi_s(u_s) < \phi_s(v_s)$. Then for $U : \{x,y\} \times \Situations \to \R$ satisfying $U(x,s) = u_s$ and $U(y,s) = v_s$, we have $x \in \widehat{\text{NWD}}(\{x,y\},\{s\},U) \subseteq \widehat{\text{NSD}}(\{x,y\},\{s\},U)$ and $x \notin \widehat{\text{NSD}}(\{x,y\},\{s\},\Phi \circ U) \supseteq \widehat{\text{NWD}}(\{x,y\},\{s\},\Phi \circ U)$, so that properties~\ref{propertynpsd1} and \ref{propertynpsd2} both fail---a contradiction.

This establishes the `only if' part of \Cref{theorem:psd}\ref{npsd1}. To establish the `only if' part of \Cref{theorem:psd}\ref{npsd2}, suppose that there is an $s \in \Situations$ such that $\phi_s$ is not strictly increasing and that there is at least one $t \in \Situations$ such that $\phi_t$ is non-constant; we must show that property~\ref{propertynpsd2} fails. 

\begin{itemize}
    
    \item Suppose first that there is a $t \in \Situations \setminus \{s\}$ such that $\phi_t$ is non-constant. Then (recalling that $\phi_s$ and $\phi_t$ are weakly increasing) there are $u>v$ and $u'>v'$ in $\R$ such that $\phi_s(u)=\phi_s(v)$ and $\phi_t(u')>\phi_t(v')$. Hence for $U : \{x,y\} \times \Situations \to \R$ satisfying $U(x,s) = u$, $U(x,t) = v'$, $U(y,s) = v$ and $U(y,t) = u'$, we have $\widehat{\text{NWD}}(\{x,y\},\{s,t\},U) \ni x \notin \widehat{\text{NWD}}(\{x,y\},\{s,t\},\Phi \circ U)$, so property~\ref{propertynpsd2} fails.
    
    \item Suppose instead that $\phi_t$ is constant for every $t \in \Situations \setminus \{s\}$. Then $\phi_s$ is (weakly increasing and) non-constant, so there are $u>v$ such that $\phi_s(u)>\phi_s(v)$. Fix arbitrary $t \in \Situations \setminus \{s\}$ and $u'>v'$ in $\R$. Then for $U : \{x,y\} \times \Situations \to \R$ satisfying $U(x,s) = v$, $U(x,t) = u'$, $U(y,s) = u$ and $U(y,t) = v'$, we have $\widehat{\text{NWD}}(\{x,y\},\{s,t\},U) \ni x \notin \widehat{\text{NWD}}(\{x,y\},\{s,t\},\Phi \circ U)$, so property~\ref{propertynpsd2} fails.
    
\end{itemize}

For the necessity (`only if') claims in \Cref{theorem:psd}\ref{npsd3}--\ref{npsd4}, assume that either property~\ref{propertynpsd3} or property~\ref{propertynpsd4} holds. Mimicking (the beginning of) the proof of \Cref{lem:sitwise_nec}\ref{lem:case-wise} in \cref{pf:maintheorem} yields that $\Phi$ must be \emph{strictly case-wise monotone} in the sense that for any $u,v \in \R^\Situations$ and any case $s \in \Situations$, $u_s > v_s$ implies $\Phi_s(u) > \Phi_s(v)$.

This establishes the `only if' part of \Cref{theorem:psd}\ref{npsd3}. To establish the `only if' part of \Cref{theorem:psd}\ref{npsd4}, suppose that there are $u,v \in \R^\Situations$ and $s,t \in \Situations$ such that $u_s > v_s$, $u_t \geq v_t$, and $\Phi_t(u) < \Phi_t(v)$; we must show that property~\ref{propertynpsd4} fails. By strict case-wise monotonicity, $\Phi_s(u) > \Phi_s(v)$ and (thus) $s \neq t$. Then for $U : \{x,y\} \times \Situations \to \R$ given by $U(x,\cdot) \coloneqq v$ and $U(y,\cdot) \coloneqq u$, we have $\widehat{\text{NWD}}(\{x,y\},\{s,t\},\Phi \circ U) \ni x \notin \widehat{\text{NWD}}(\{x,y\},\{s,t\},U)$, so property~\ref{propertynpsd4} fails.
\qed

\section{Proof of the claim in \texorpdfstring{\Cref{remark:rich}}{Remark~\ref{remark:rich}}}
\label{app:pf_remark:rich}

We will exhibit a set $\Situations$ with $2 \leq \abs*{\Situations} < \infty$ and a map $\Phi : \Z^\Situations \to \Z^\Situations$ that is not case-wise, and yet satisfies properties~\ref{item:shrink1} and \ref{item:shrink2} when restricted to finite problems with integer-valued payoffs, i.e. satisfies the properties
\begin{enumerate}[label=($\arabic*_\Z$)]
\setcounter{enumi}{2}
\item \label{item:shrink1_Z} $\text{NSD}(\objects,\situations,\Phi \circ U) \subseteq \text{NSD}(\objects,\situations,U)$ for every finite problem $(\objects,\situations,U)$ such that $U(x,\cdot) \in \Z^\Situations$ for each $x \in \objects$.
\item \label{item:shrink2_Z} $\text{NWD}(\objects,\situations,\Phi \circ U) \subseteq \text{NWD}(\objects,\situations,U)$ for every finite problem $(\objects,\situations,U)$ such that $U(x,\cdot) \in \Z^\Situations$ for each $x \in \objects$.
\end{enumerate}
To that end, let $\Situations \coloneqq \{1,2\}$, and define $\Phi : \Z^2 \to \Z^2$ by
\begin{equation*}
    \Phi(u_1,u_2) \coloneqq 
    \begin{cases}
        (u_1,u_2) & \text{if $u_1 < 0$ or if $u_1 = 0$ and $u_2 < 0$}\\
        (1,u_2) & \text{if $u_1 = 0$ and $u_2 \ge 0$}\\
        (3u_1,u_2) & \text{if $u_1 > 0$.}
    \end{cases}
\end{equation*}
Clearly $\Phi$ is not case-wise: there do not exist functions $\phi_1,\phi_2 : \R \to \R$ such that $\Phi(u_1,u_2) = (\phi_1(u_1),\phi_2(u_2))$ for every $(u_1,u_2) \in \Z^2$.

For any finite problem $(\objects,\situations,U)$ and any $x \in \objects$, write
\begin{equation*}
    Y_x(\objects,\situations,U) \coloneqq \{ y \in \objects: \text{$U(y,s) > U(x,s)$ for some $s \in \situations$} \} .
\end{equation*}
To show that properties~\ref{item:shrink1_Z} and \ref{item:shrink2_Z} hold, we shall prove that
\begin{enumerate}[label=($\lozenge$)]
\item \label{diamondprop} for any finite problem $(\objects,\situations,U)$ such that $U(y,\cdot) \in \Z^2$ for each $y \in \objects$, any $x\in \objects$, any $p \in \Delta(Y_x(\objects,\situations,U))$, and any $s \in \situations$,
\begin{align*}
    \sum_{y \in \objects} U(y,s) p(y) &\underset{(>)}{\geq} U(x,s)
    \\
    \text{implies} \quad
    \sum_{y \in \objects} \Phi_s(U(y,\cdot)) p(y) &\underset{(>)}{\geq} \Phi_s(U(x,\cdot)) .
\end{align*}
\end{enumerate}
This implies property~\ref{item:shrink1_Z} because for any finite problem $(\objects,\situations,U)$ such that $U(y,\cdot) \in \Z^2$ for each $y \in \objects$ and any $x \in \objects \setminus \text{NSD}(\objects,\situations,U)$, there exists a $p \in \Delta(Y_x(\objects,\situations,U))$ such that  
\begin{equation*}
    \sum_{y \in \objects} U(y,s) p(y) > U(x,s) \quad \text{for every $s \in \situations$,}
\end{equation*}
whence
\begin{equation*}
    \sum_{y \in \objects} \Phi_s(U(y,\cdot)) p(y) > \Phi_s(U(x,\cdot)) \quad \text{for every $s \in \situations$}
\end{equation*}
by property~\ref{diamondprop}, so that $x \notin \text{NSD}(\objects,\situations,\Phi \circ U)$. Similarly, property~\ref{diamondprop} implies property~\ref{item:shrink2_Z} since if $x \in \objects \setminus \text{NWD}(\objects,\situations,U)$ then there is a $p \in \Delta(Y_x(\objects,\situations,U))$ such that  
\begin{equation*}
\sum_{y \in \objects} U(y,s) p(y) \underset{(>)}{\geq} U(x,s) \quad \text{for $\underset{\text{(some)}}{\text{every}}$ $s \in \situations$,}
\end{equation*}
whence
\begin{equation*}
\sum_{y \in \objects} \Phi_s(U(y,\cdot)) p(y) \underset{(>)}{\geq} \Phi_s(U(x,\cdot)) \quad \text{for $\underset{\text{(some)}}{\text{every}}$ $s \in \situations$}
\end{equation*}
by property~\ref{diamondprop}, and thus $x \notin \text{NWD}(\objects,\situations,\Phi \circ U)$.

To establish that property~\ref{diamondprop} holds, fix a finite problem $(\objects,\situations,U)$ such that $U(y,\cdot) \in \Z^2$ for each $y \in \objects$, and further fix $x\in \objects$, $p \in \Delta(Y_x(\objects,\situations,U))$, and $s \in \situations$ such that 
\begin{equation*}
    \sum_{y \in \objects} U(y,s) p(y) \underset{(>)}{\geq} U(x,s) .
\end{equation*}
It suffices to exhibit a convex and strictly increasing $\phi : \R \to \R$ such that $\Phi_s(U(x,\cdot)) = \phi(U(x,s))$ and  $\Phi_s(U(y,\cdot)) \ge \phi(U(y,s))$ for all $y \in Y_x(\objects,\situations,U)$, since then
\begin{equation*}
    \sum_{y \in \objects} \Phi_s(U(y,\cdot)) p(y)
    \ge \sum_{y \in \objects} \phi(U(y,s)) p(y) \underset{(>)}{\geq} \phi(U(x,s)) = \Phi_s(U(x,\cdot))
\end{equation*}
by \hyperref[theorem:pratt]{Pratt's theorem}. If $s=2$, then the identity ($\phi$ given by $\phi(u) \coloneqq u$ for each $u \in \R$) has this property, since $\Phi_2(u_1,u_2) = u_2$ for all $(u_1,u_2) \in \Z^2$.

Suppose instead that $s=1$. If either $U(x,1) \ne 0$ or $U(x,2) < 0$, then let $\phi : \R \to \R$ be given by
\begin{equation*}
    \phi(u) \coloneqq 
    \begin{cases}
        u & \text{if $u \le 0$}\\
        3u & \text{if $u > 0$,}
    \end{cases}
\end{equation*}
and note that $\phi$ has the desired properties since $\Phi_1(u_1,u_2) \ge \phi(u_1)$ for all $(u_1,u_2) \in \Z^2$, where the inequality is strict only if $u_1 = 0$ and $u_2 \ge 0$. If instead $U(x,1) = 0$ and $U(x,2) \ge 0$, then let $\phi : \R \to \R$ be given by 
\begin{equation*}
    \phi(u) \coloneqq 
    \begin{cases}
        u & \text{if $u \le -1$}\\
        2u+1 & \text{if $-1 < u < 1$}\\
        3u & \text{if $u \ge 1$,}
    \end{cases}
\end{equation*}
and note that $\phi$ has the desired properties since $\phi(U(x,1)) = 1 = \Phi_1(U(x,\cdot))$ and, writing
\begin{equation*}
B \coloneqq \left\{ (v_1,v_2) \in \Z^2 : \text{either $v_1 \neq 0$ or $v_2 \geq 0$} \right\} ,
\end{equation*}
we have $B \supseteq \{ U(y,\cdot) : y \in Y_x(\objects,\situations,U) \}$ and $\Phi_1(u_1,u_2) = \phi(u_1)$ for all $(u_1,u_2) \in B$.
\qed

\end{appendices}



\printbibliography[heading=bibintoc]

\end{document}